\documentclass[journal,10pt,twocolumn,twoside]{IEEEtran}

\IEEEoverridecommandlockouts
\usepackage{amsmath}
\usepackage{amsfonts}
\usepackage{cases}
\usepackage{setspace}
\usepackage{fancybox}
\usepackage{subfigure}
\usepackage{epsfig}
\usepackage{graphicx}
\usepackage{epstopdf}
\usepackage{float}
\usepackage{multirow}
\usepackage{color}
\usepackage{amsmath}
\usepackage{multirow}
\usepackage{indentfirst}
\usepackage{dsfont}
\usepackage{amsfonts}
\usepackage{times,amsmath,color,amssymb,epsfig,cite,subfigure,algorithm,algorithmic}

\newenvironment{proof}[1][Proof]{\begin{trivlist}
\item[\hskip \labelsep {\bfseries #1}]}{\end{trivlist}}

\newtheorem{lemma}{Lemma}
\newtheorem{theorem}{Theorem}

\usepackage{color}

\begin{document}
\title{Designing Massive MIMO Detector via PS-ADMM approach}
\author{Quan~Zhang,
        Yongchao~Wang,~\IEEEmembership{Senior~Member,~IEEE}
}

\markboth{}
{}

\maketitle

\begin{abstract}
In this paper, we develop an efficient detector for massive multiple-input multiple-output (MIMO) communication systems via penalty-sharing alternating direction method of multipliers (PS-ADMM). Its main content are as follows:
first, we formulate the MIMO detection as a maximum-likelihood optimization problem with bound relaxation constraints.
Then, the higher modulation signals are decomposed into a sum of multiple binary variables through their inherent structures, by exploiting introduced binary variables as penalty functions, the detection optimization model is equivalent to a nonconvex sharing minimization problem.
Second, a customized ADMM algorithm is presented to solve the formulated nonconvex optimization problem. In the implementation, all variables can be solved analytically and parallelly.
Third, it is proved that the proposed PS-ADMM algorithm converges if proper parameters are chosen.
Simulation results demonstrate the effectiveness of the proposed approach.
\end{abstract}


\IEEEpeerreviewmaketitle

\section{Introduction}
\IEEEPARstart{M}{assive} multiple-input multiple-output (MIMO) technology, which invokes large number of antennas are equipped at the base station (BS) and serve a large number of user terminals in the same frequency band, is widely considered to be one of the  disruptive technologies of fifth-generation (5G) communication systems \cite{marzetta2010noncooperative}, \cite{Boccardi2013Five}. The foremost benefit of massive MIMO is the significant increase in the spatial degrees of freedom that can helps improve the throughput and energy efficiency by several orders of magnitude over conventional MIMO systems \cite{Ngo2013Energy}.
However, numerous practical challenges arise in implementing massive MIMO technology in order to achieve such improvements.
One such challenge is signal detection lies in uplink for a massive system, for which it is difficult to achieve an effective compromise among good detecting performance, low computational complexity and high processing  parallelism \cite{albreem2019massive}.

The optimal MIMO detection is maximum-likelihood (ML) detector \cite{Verdu1986Minimum}, suffers from an exponential increase in computational complexity with an increasing number of terminal antennas, which entails prohibitive complexity for the massive MIMO detection \cite{albreem2019massive}. Various nonlinear detection methods such as sphere decoding (SD) \cite{damen2003maximum}, the semidefinite relaxation (SDR) \cite{Zhi2010Semidefinite}, the PSK detector (PSKD) \cite{luo2003efficient}, K-best \cite{guo2006algorithm}, and the triangular approximate semidefinite relaxation (TASER) \cite{castaneda2016data},
have achieved near-optimal ML detection performance in small-scale MIMO systems \cite{yang2015fifty}, but they are still prohibitively complex for large-scale or high-order modulation MIMO systems.
Linear detection methods such as minimum mean square error (MMSE) \cite{6771364} and zero-forcing (ZF) \cite{1092893}, are one of the best choices with respect to (w.r.t) the tradeoff in performance and complexity, especially when the BS-to-user-antenna ratio is large \cite{rusek2012scaling}. In order to further reduce computational complexity and improve the detecting performance,
on the one hand, numerous detectors can be classified as approximate matrix inversion methods have been proposed to reduce the computational complexity of linear detectors \cite{albreem2019massive}, such as Neumann series (NS) \cite{wu2014large}, \cite{vcirkic2014complexity}, Gauss-Seidel (GS) \cite{dai2015low}, \cite{wu2016efficient}, Richardson (RI) \cite{gao2014low} and conjugate gradient (CG) \cite{yin2014conjugate} methods. However, the decrease in computational complexity of these algorithms comes at the expense of the loss of detecting performance, meanwhile these methods also deliver a poor BER performance when the BS-to-user antenna ratios is close to one.
On the other hand, various detection algorithms have been proposed to obtain better detecting performance than that of linear detection algorithms, which include a class of algorithms named BOX detection \cite{albreem2019massive}. In \cite{wu2016high}, a optimized coordinate descent with box-constrained equalization (OCD-BOX) shows better bit error rate (BER) performance with low hardware complexity, however, it cann't be implemented parallelly. In \cite{shahabuddin2017admm}, a detection algorithm based on alternating direction method of multipliers (ADMM) with infinity norm or box-constrained equalization named ADMIN has been proposed, which outperforms linear detectors by a large margin if the BS-to-user-antenna ratio is small.

In recent years, ADMM technique was widely used to solve convex and nonconvex problems due to its simplicity, operator splitting capabilities and convergence guarantees under mild conditions \cite{boyd2011distributed}. In \cite{7526551}, ADMM was introduced into the MIMO detection area as a example to illustrate that the ADMM method is very effective in solving mixed-integer quadratic programming.
The ADMM detection algorithm was applied in various scenarios for MIMO systems \cite{souto2016mimo}, \cite{souto2018efficient} and improved for massive MIMO systems \cite{lopes2018iterative,elgabli2019a,elgabli2019pro,feng2019a,8755566}. Although ADMM-based methods above-mentioned can achieve BER performance significantly outperform than that of conventional detectors, there are two major deficiencies severely limit the high accuracy solution of detection optimization problem be solved. First, the constraints set of optimization problem is over-relaxed rather than exact replaced. Second, these works haven't a specifical structure of the objective in optimization problem for high-order modulation systems  .

In this paper, we focus on designing a new ADMM-based detector for massive MIMO systems. By combining ideas of penalized bound relaxation for the ML detection formulation in \cite{ma2002quasi}, high-order QAM signals are converted into binary forms in \cite{mao2007semidefinite} and sharing-ADMM technique in \cite{boyd2011distributed}, \cite{hong2016convergence}, respectively, we obtain a new detector for massive MIMO systems, called PS-ADMM, which has favorable BER performance meanwhile providing a cheap complexity. The main technical contributions of this paper are summarized as follows:

\begin{itemize}
\item \emph{Penalty sharing formulation}:
the MIMO detection problem is formulated as a ML optimization problem with bound relaxation constraints. Then, high-order modulation signals are decomposed into a sum of multiple binary variables, by introducing these binary variables as penalty functions, the ML problem is equivalent to a nonconvex sharing minimization problem.
\item \emph{Efficient implementation}: the main advantage of the formulated penalty sharing model is it can be solved by sharing-ADMM algorithms. In the favourable execution architecture, all the variables in subproblems can be solved analytically, accurately and updated parallelly in each iteration step. As a result, the PS-ADMM algorithm achieves excellent BER performance while cheap computational complexity providing, especially when the BS-to-user-antenna ratio is close to one.
\item \emph{Theoretically-guaranteed performance}:
we prove that the proposed PS-ADMM algorithm is convergent and can approach arbitrarily close to a stationary point of the nonconvex optimization problem if proper parameters are chosen.
This also means that the updates will not change significantly from their initial values, the iterative method can be start at zero or some other default initialization.
\end{itemize}

The rest of this paper is organized as follows.
In Section \ref{sec:problem-formulation}, we formulate the massive MIMO detection problem to a nonconvex sharing minimization problem.
In Section \ref{sec:PS-ADMM Solving Algorithm}, an efficient sharing-ADMM algorithm is customized to solve the nonconvex minimization problem.
Section \ref{sec:Analysis-psadmm} presents the detailed performance analysis, including convergence and complexity of the proposed PS-ADMM algorithm.
Simulation results, which show the effectiveness of our proposed PS-ADMM algorithm, are presented in Section \ref{sec:Simulation results}
and the conclusions are given in Section \ref{sec:Conclusion}.

\emph{Notations}: In this paper, bold lowercase, uppercase and italics letters denote column vectors, matrices and scalars respectively;
 $\mathds{C}$ denotes the the complex field;
 $(\cdot)^H$ symbolizes the conjugate transpose operation;
 $x_{iR}$, $x_{iR}$ denote the real and imaginary parts of the $ith$ entry of a vector $\mathbf{x}$ respectively;
 $\|\cdot\|_2$ represents the 2-norm of vector $\mathbf{x}$;
 $\underset{[a,b]}\Pi(\cdot)$ denotes the Euclidean projection operator onto the interval $[a,b]$;
 $\nabla(\cdot)$ represents the gradient of a function;
 $\rm{Re}(\cdot)$ takes the real part of the complex variable;
 $\lambda_{\rm min}(\cdot)$ and $\lambda_{\rm max}(\cdot)$ denote the minimum and maximum eigenvalue of a matrix respectively;
 $\langle \cdot,\cdot\rangle$ denotes the dot product operator and ${\mathbf{I}}$ denotes an identity matrix.

\section{System Model And Problem Formulation}
\label{sec:problem-formulation}

The considered signal detection problem lies in uplink of the massive multi-user (MU) MIMO systems, where BS equipped with $B$ antennas serves $U$ single-antenna users. Here, we assume $ B \ge U$. Typically, the received signal vector at BS can be characterized by the following model
\begin{equation}\label{transmission model}
{\mathbf{r}} = {\mathbf{Hx}} + {\mathbf{n}},
\end{equation}
in \eqref{transmission model}, $\mathbf{x}\in \mathcal{X}^{U}$ is the transmitted signal vector from the U users and $\mathcal{X}$ refers to the signal constellation set, $\mathbf{r}\in \mathbb{C}^{B}$ is the BS received signal vector, $\mathbf{H}\in \mathbb{C}^{B\times U}$ denotes the MIMO channel matrix, and $\mathbf{n}\in \mathbb{C}^{B}$ denotes additive white Gaussian noise. The entries of $\mathbf{H}$ and $\mathbf{n}$ are assumed to be independent and identically distributed (i.i.d.) complex Gaussian variables with zero mean.

The ML detector for MIMO QAM signals, i.e., achieving minimum error probability of detecting $\mathbf{x}$ from the received signal $\mathbf{r}$, can be formulated as the following discrete least square problem \cite{Zhi2010Semidefinite}
\begin{equation}\label{eq:MLdetection}
\begin{split}
&\min_{\mathbf{x}\in \mathcal{X}^U} \Vert\mathbf{r}-\mathbf{H}\mathbf{x} \Vert_2^{2}, \\
\end{split}
\end{equation}
where ${\mathcal X}= \{x= x_{R} + j x_{I} \vert x_{R}, x_{I} \in \{\pm 1, \pm 3, \cdots , \pm (2^{Q} -1)\}\}$ and $Q$ is some positive integer. The model \eqref{eq:MLdetection} is a typical combination optimization problem \cite{verdu1989computational} since the constraint $\mathbf{x}\in \mathcal{X}^{U}$ is discrete. It means that obtaining \eqref{eq:MLdetection}'s global optimal solution is prohibitive in practice since the corresponding computational complexity grows exponentially with the users' number $U$, BS's antenna number $B$, and the set $\mathcal{X}$'s size \cite{chen2017a}. In the following, through exploiting insight structures of the model \eqref{eq:MLdetection}, we proposed a {\it relaxation-tighten} technique and transform it to the well-known sharing problem.

Let $\mathbf{x}_q = \mathbf{x}_{qR}+j\mathbf{x}_{qI}$, where $\mathbf{x}_q\in \mathcal{X}^U_q = \{\mathbf{x}_{qR} + j \mathbf{x}_{qI}| \mathbf{x}_{qR}, \mathbf{x}_{qI} \in \{1,-1\}^U\}$ and $q = 1,\cdots,Q$. Then, any transmitted signal vector $\mathbf{x}\in\mathcal{X}^{U}$ can be always expressed by
\begin{equation}\label{eq:XsumXq}
\mathbf{x}=\sum_{q=1}^{Q}2^{q-1} \mathbf{x}_q,\ q = 1,\cdots,Q.
\end{equation}
Plugging \eqref{eq:XsumXq} into the model \eqref{eq:MLdetection}, it can be equivalent to
\begin{subequations}\label{eq:S_MLdetection}
\begin{align}
&\hspace{0.4cm}\min_{\mathbf{x}_q} \hspace{0.2cm} \frac{1}{2}\Vert\mathbf{r}-\mathbf{H}(\sum_{q=1}^{Q}2^{q-1} \mathbf{x}_q) \Vert_2 ^{2},  \label{eq:S_MLdetection_a} \\
&\hspace{0.5cm}{\rm {s.t.}}   \hspace{0.3cm}\mathbf{x}_q \in \mathcal{X}_q^{U},\ q = 1,\cdots,Q. \label{eq:S_MLdetection_b}
\end{align}
\end{subequations}
Relax the binary integer constraints \eqref{eq:S_MLdetection_b} to the box constraints $\mathbf{x}_q\in\tilde{\mathcal{X}}^U_q= \{\mathbf{x}_q= \mathbf{x}_{qR} + j \mathbf{x}_{qI} \vert\mathbf{x}_{qR}, \mathbf{x}_{qI} \in [-1\ 1]^U\}$ and then tighten the relaxation by adding the quadratic penalty function into the objective \eqref{eq:S_MLdetection_a}. We can transform the model \eqref{eq:S_MLdetection} to
\begin{equation}\label{eq:PS_ML}
\begin{split}
&\hspace{0.3cm} \min_{\mathbf{x}_q} \hspace{0.2cm} \frac{1}{2}\Vert\mathbf{r}-\mathbf{H}(\sum_{q=1}^{Q}2^{q-1} \mathbf{x}_q) \Vert_2 ^{2}-\sum_{q=1}^{Q}\frac{\alpha_q}{2}\Vert\mathbf{x}_q \Vert_2 ^{2}\\
&\hspace{0.4cm} {\rm {s.t.}} \hspace{0.4cm}\mathbf{x}_q\in \tilde{\mathcal{X}}_q^{U},\;q = 1,\cdots,Q,
\end{split}
\end{equation}
where penalty parameters $\alpha_q \ge 0$. It is easy to see that the non-convex quadratic penalty function can make the integer solutions more favorable. Moreover, by introducing auxiliary variable $\mathbf{x}_0\in \mathbb{C}^{U}$, the model \eqref{eq:PS_ML} can be further equivalent to
\begin{equation}\label{eq:PS_ADMM}
\begin{split}
&\min_{\mathbf{x}_0, \mathbf{x}_q} \hspace{0.2cm} \frac{1}{2}\Vert\mathbf{r}-\mathbf{H}\mathbf{x}_0 \Vert_2 ^{2}-\sum_{q=1}^{Q}\frac{\alpha_q}{2}\Vert\mathbf{x}_q \Vert_2 ^{2}\\
&\ \ {\rm {s.t.}} \hspace{0.3cm} \mathbf{x}_0=\sum_{q=1}^{Q}2^{q-1} \mathbf{x}_q,\ \mathbf{x}_q\in  \tilde{\mathcal{X}}_q^{U}, q = 1,\cdots,Q,
\end{split}
\end{equation}
which is called by sharing problem. In the following,  an efficient ADMM solving algorithm, named PS-ADMM, to solve \eqref{eq:PS_ADMM}. We also provide detailed analysis of the proposed PS-ADMM algorithm on convergence, computational complexity, and solution's quality.

\section{PS-ADMM Solving Algorithm}
\label{sec:PS-ADMM Solving Algorithm}

ADMM is a popular and powerful technique for solving large-scale  optimization problems \cite{boyd2011distributed}.
In this section, we show a new MIMO detection algorithm, through exploiting ADMM technique, to solve \eqref{eq:PS_ADMM}.

The augmented Lagrangian function of the problem \eqref{eq:PS_ADMM} can be expressed as
\begin{align}
\begin{split}\label{eq:lagrangian_PSADMM}
&L_{\rho}(\{\mathbf{x}_q,q\in \mathcal Q\},\mathbf{x}_0,\mathbf{y}) = \frac{1}{2}\Vert\mathbf{r}\!-\!\mathbf{H}\mathbf{x}_0 \Vert_2 ^{2}\!-\!\sum_{q=1}^{Q}\frac{\alpha_q}{2}\!\Vert\mathbf{x}_q\!\Vert_2 ^{2}\\
&\ \ +{\rm Re}\big\langle \mathbf{x}_0-\sum_{q=1}^{Q}2^{q-1} \mathbf{x}_q, \mathbf{y}\big\rangle+\frac{\rho}{2}\big\|\mathbf{x}_0-\sum_{q=1}^{Q}2^{q-1}\mathbf{x}_q \big\|_2^2,
\end{split}
\end{align}
where $\mathbf{y}\in \mathbb{C}^{U}$ and $\rho>0$ are the Lagrangian multiplier and penalty parameter respectively.
Based on the above augmented Lagrangian, the classical algorithm framework to solve \eqref{eq:PS_ADMM} can be described as
\begin{subequations}\label{PSADMM_update_ori}
\begin{align}
&\mathbf{x}_q^{k+1}\!\!=\!\! \mathop{\arg\min}_{\mathbf{x}_q\in\tilde{\mathcal{X}}_q^{U}}\! L(\mathbf{x}_q,\!\mathbf{x}_1^{k+1}\!\!\!\!,\!\cdots\!,\mathbf{x}_{q-1}^{k+1}\!, \mathbf{x}_{q+1}^k\!,\!\cdots\!,\!\mathbf{x}_Q^k,\!\mathbf{x}_0^{k},\! \mathbf{y}^{k}),\nonumber\\
& \hspace{5.8cm} q=1,\cdots,Q,  \label{eq:x_q_update}\\
&\mathbf{x}_0^{k+1}=\arg\min_{\mathbf{x}_0} \;\;L_{\rho}( \{\mathbf{x}_q^{k+1}\}_{q=1}^{Q}, \mathbf{x}_0, \mathbf{y}^{k}), \label{eq:x0_update}\\
&\mathbf{y}^{k+1}=\mathbf{y}^k+\rho\Big(\mathbf{x}_0^{k+1}-\sum_{q=1}^{Q}2^{q-1} \mathbf{x}_q^{k+1}\Big),\label{eq:y_update}
\end{align}
\end{subequations}
where $k$ denotes the iteration number.

The main challenge of implementing \eqref{PSADMM_update_ori} lies in how to solve suboptimization problems \eqref{eq:x_q_update} and \eqref{eq:x0_update} efficiently.
For \eqref{eq:x_q_update}, it can be observed that $L_{\rho}(\{\mathbf{x}_q\}_{q=1}^{Q},
\mathbf{x}_0^k, \mathbf{y}^k)$ is a strongly convex quadratic function with respect to some specific $\mathbf{x}_q$ when $4^{q-1}\rho>\alpha_q$. It means that the solution of the suboptimization problems \eqref{eq:x_q_update} can be obtained through the following procedures: set the gradient of the corresponding augmented Lagrangian function with respect to $\mathbf{x}_q$ to be zero, i.e.,
\begin{equation}
\nabla_{\mathbf{x}_q}L_{\rho}(\mathbf{x}_q,\!\mathbf{x}_1^{k+1}\!\!,\!\cdots\!,\mathbf{x}_{q-1}^{k+1}\!, \mathbf{x}_{q+1}^k,\cdots,\mathbf{x}_Q^k,\mathbf{x}_0^{k},\!\mathbf{y}^{k}) \!= \!0,
\end{equation}
which leads to the following linear equation
\begin{equation}\label{eq:x_q_update_zero}
\begin{split}
& \nabla_{\mathbf{x}_q} \Big(-\frac{\alpha_q}{2}\Vert\mathbf{x}_q \Vert_2 ^{2}-{\rm Re}\langle 2^{q-1}\mathbf{x}_q,\;\mathbf{y}^k\rangle  \\
&\!+\!\frac{\rho}{2}\|\mathbf{x}_0^{k}\!-\!\!\sum_{i<q} 2^{i-1}\mathbf{x}_i^{k+1}\!-\!\!\sum_{i>q}2^{i-1} \mathbf{x}_i^{k}\!-\!2^{q-1} \mathbf{x}_q \|_2^2 \Big)\!=\!0.
\end{split}
\end{equation}
Noticing the variables in $\mathbf{x}_q$ are separable, we can obtain global optimal solution of the problem \eqref{eq:x_q_update} in the following
\begin{equation}\label{eq: solution_Xq}
\begin{split}
&\mathbf{x}_q^{k+1} = \underset{[-1,1]}\Pi \bigg(\frac{2^{q-1}}{4^{q-1}\rho-\alpha_q)}\Big( \rho\mathbf{x}_0^{k}-\rho\sum_{i<q} 2^{i-1}\mathbf{x}_i^{k+1} \\
&\hspace{0.9cm} -\rho\sum_{i>q}2^{i-1} \mathbf{x}_i^{k}+ \mathbf{y}^k \Big) \bigg),\ \ \ q = 1,\cdots,Q,
\end{split}
\end{equation}
where $\underset{[-1,1]}\Pi(\cdot)$ performs the following operation: project every entry's real part or imaginary part of the input vector onto [-1,1].

Moreover, $L_{\rho}(\{\mathbf{x}_q^{k+1}\}_{q=1}^{Q}, \mathbf{x}_0, \mathbf{y}^k)$ is a strongly convex quadratic function since $\rho>0$ and matrix $\mathbf{H}^{H} \mathbf{H} $ is positive definite. Then, the optimal solution of the suboptimization problem \eqref{eq:x0_update} can be obtained through setting $\nabla_{\mathbf{x}_0} L_{\rho}(\{\mathbf{x}_q^{k+1}\}_{q=1}^{Q}, \mathbf{x}_0, \mathbf{y}^{k}) =\!0$ to be zero and solving the corresponding linear equation, which results in
\begin{equation}\label{eq: solution X0}
\mathbf{x}_0^{k+1} = {(\mathbf{H}^{H} \mathbf{H}+\rho\mathbf{I})}^{-1}\bigg(\! \mathbf{H}^{H} \mathbf{r} + \rho\sum_{q=1}^{Q}2^{q-1} \mathbf{x}^{k+1}_q - \mathbf{y}^{k}\!\bigg).
\end{equation}

 To be clear, we summarize the proposed PS-ADMM algorithm for solving model \eqref{eq:PS_ADMM} in \emph{Algorithm \ref{PS-ADMM algorithm}}.

\begin{algorithm}[!t]
\caption{The proposed PS-ADMM algorithm}
\label{PS-ADMM algorithm}
\begin{algorithmic}[1]
\REQUIRE $\mathbf{H}$, $\mathbf{r}$, ${Q}$, $\rho$, $\{\alpha_q\}_{q=1}^{Q}$\\
\ENSURE $\mathbf{x}_0^k$\\
\STATE Initialize $\{\mathbf{x}_q^{1}\}_{q=1}^{Q}, \mathbf{x}_0^{1}, \mathbf{y}^1$ as the all-zeros vectors\footnotemark.
\STATE  \textbf{For $k = 1,2,\cdots$}
\STATE \hspace{0.2cm} Step 1: Update $\{\mathbf{x}_q^{k+1}\}_{q=1}^{Q}$, sequentially via \eqref{eq: solution_Xq}.
\STATE \hspace{0.2cm} Step 2: Update $\mathbf{x}_0^{k+1}$ via \eqref{eq: solution X0}.
\STATE \hspace{0.2cm} Step 3: Update $\mathbf{y}^{k+1}$ via \eqref{eq:y_update}.
\STATE \textbf{Until} some preset condtion is satisfied.
\end{algorithmic}
\end{algorithm}
\footnotetext{Unlike some existing algorithms that are derived from the solution of the simple linear algorithm such as MMSE or ZF, there is no computations required for PS-ADMM to perform initialization, and the initial values $\{\mathbf{x}_q^{1}\}_{q=1}^{Q}$, $\mathbf{x}_0^{1}$, $\mathbf{y}^{1}$ can be set to zeros, ones, minus ones and random values.}
\vspace{-0.1cm}

\section{Performance Analysis}\label{sec:Analysis-psadmm}
In this section we make a detailed analysis of \emph{Algorithm \ref{PS-ADMM algorithm}} from the viewpoints of convergence property, convergence rate and computational complexity.
\subsection{Convergence property}
We have the following theorem to show convergence properties of the proposed PS-ADMM algorithm.
\begin{theorem}\label{thm:convergence}
Assume parameters $\rho$ and $\{\alpha_q\}_{q=1}^{Q}$ satisfy $4^{q-1}\rho>\alpha_q$ and $\rho > \sqrt{2} \lambda_{\rm max}(\mathbf{H}^H\mathbf{H})$, where q = 1,···,Q. The sequence $\{\{\mathbf{x}^{k}_q\}_{q=1}^{Q}, \mathbf{x}_0^{k}, \mathbf{y}^{k}\}$ generated by \emph{Algorithm \ref{PS-ADMM algorithm}} is convergent, i.e.,
\begin{equation}\label{convergence variables}
\begin{split}
&\lim\limits_{k\rightarrow+\infty}\mathbf{x}^{k}_q=\mathbf{x}^*_q, \ \ \lim\limits_{k\rightarrow+\infty}\mathbf{x}_0^{k}=\mathbf{x}_0^*, \lim\limits_{k\rightarrow+\infty}\mathbf{y}^{k}=\mathbf{y}^*, \\
& \hspace{0.3cm} \forall~\mathbf{x}_q\in \tilde{\mathcal{X}}_q^{U}, \; q = 1,\cdots,Q.
\end{split}
\end{equation}
Moreover, $\{\mathbf{x}^*_q\}_{q=1}^{Q}$ is a stationary point of original problem \eqref{eq:PS_ML}, i.e., it satisfies the following inequality

\begin{align}\label{stationary point}
& {\rm Re}\Big \langle\nabla_{\mathbf{x}_q} \Big(\ell\big(\sum_{q=1}^{Q}2^{q-1} \mathbf{x}_q^*\big)-\sum_{q=1}^{Q}\frac{\alpha_q}{2}\Vert\mathbf{x}_q^* \Vert_2 ^{2}\Big),\mathbf{x}_q-\mathbf{x}^*_q\Big\rangle\nonumber\\
& \quad\quad\ge 0,\hspace{1.8cm}\forall~\mathbf{x}_q\in \tilde{\mathcal{X}}_q^{U},\; q = 1,\cdots,Q.
\end{align}
\end{theorem}
where $\ell\left(\sum_{q=1}^{Q}2^{q-1} \mathbf{x}_q^*\right)=\frac{1}{2}\Vert\mathbf{r}-\mathbf{H}(\sum_{q=1}^{Q}2^{q-1} \mathbf{x}_q^*) \Vert_2 ^{2}$, we use $\ell\left(\mathbf{x}\right)$ to denote $\frac{1}{2}\Vert\mathbf{r}-\mathbf{H}\mathbf{x} \Vert_2 ^{2}$ in the rest of this paper.

{\it Remarks:}
   Theorem \ref{thm:convergence} indicates that the proposed PS-ADMM algorithm is theoretically-guaranteed convergent to some stationary point of model \eqref{eq:PS_ML} under the conditions $4^{q-1}\rho>\alpha_q, q = 1,\cdots,Q$ and $\rho>\sqrt{2}\lambda_{\rm max}(\mathbf{H}^H\mathbf{H})$. Here, we should note that these conditions are easily satisfied since the value of penalty parameters $\rho$ and $\{\alpha_q\}_{q=1}^{Q}$ can be set accordingly when the channel matrix $\mathbf{H}$ is known. The key idea of proving Theorem \ref{thm:convergence} is to find out that potential function $L_{\rho}(\{\mathbf{x}_q\}_{q=1}^{Q}, \mathbf{x}_0, \mathbf{y})$ {\it decreases sufficiently} in every ADMM iteration and is lower-bounded. To reach this goal, we first prove several related lemmas in Appendix \ref{lemma1-3}. Then, we give the detailed proof of Theorem \ref{thm:convergence} in Appendix \ref{PS-ADMM Proof}.

\subsection{Convergence rate}
We use the residual error which is defined as $\sum_{q=1}^{Q}\|\mathbf{x}_q^{k+1}-\mathbf{x}_{q}^{k}\|_2^2 + \|\mathbf{x}_0^{k+1}-\mathbf{x}_0^{k}\|_2^2$ to measure the convergence progress of the PS-ADMM algorithm since it converges to zero as $k\rightarrow+\infty$. Then, we have Theorem \ref{iteration complexity} about its convergence progress. The detailed proof is shown in Appendix \ref{Iteration complexity Proof}.

\begin{theorem}\label{iteration complexity}
  Let $t$ be the minimum iteration index such that $\sum_{q=1}^{Q}\|\mathbf{x}_q^{k+1}-\mathbf{x}_{q}^{k}\|_2^2 + \|\mathbf{x}_0^{k+1}-\mathbf{x}_0^{k}\|_2^2\leq\epsilon$, where $\epsilon$ is the desired precise parameter for the solution. Then, we have the following iteration complexity result
 \[
   \begin{split}
     t \!\leq \!\frac{1}{C\epsilon}\bigg(\!L_{\rho}(\{\mathbf{x}^{1}_q\}_{q=1}^{Q}, \mathbf{x}_0^{1}, \mathbf{y}^{1}) \!-\! \Big(\ell\left(\mathbf{x}^*\right) \!-\! \sum_{q=1}^{Q}\frac{\alpha_q}{2} \Vert\mathbf{x}_q^* \Vert_2 ^{2}\Big)\!\bigg),
   \end{split}
 \]
where the constant \[ C\!=\!\min\!\left\{\!\{ \frac{\gamma_q(\rho)}{2}\}_{q=1}^{Q}, \Big(\frac{\gamma(\rho)}{2}\!\!-\!\!\frac{\lambda_{\rm max}^2(\mathbf{H}^{H} \mathbf{H})}{\rho}\Big)\! \right\}.\]
\end{theorem}

\subsection{Computational complexity\protect\footnotemark}\label{complexity-sec}

  The overall computational complexity of the PS-ADMM detection algorithm consists of two parts: the first part, which is independent of the number of iterations, is required to compute the $\mathbf{x}_0 $ update of PS-ADMM in \eqref{eq: solution X0}, it is needs to be calculated only once when detecting each transmitted symbol vector. The first part of calculations is performed in three steps: first, the multiplication of the $U\times B$ matrix $\mathbf{H}^{H}$ by the $B\times U$ matrix $\mathbf{H}$; second, the computation of the inversion of the regularized Gramian matrix $\mathbf{H}^{H} \mathbf{H}+\rho\mathbf{I}$; and third, the computation of the  $U\times B$ matrix $\mathbf{H}^{H}$ by the $B\times 1$ vector $\mathbf{y}$ to obtain matched-filter vector $\mathbf{H}^{H}\mathbf{y}$. These steps require $\frac{1}{2}BU^2$, $\frac{1}{3}U^3$ and $BU$ complex multiplications, respectively.
  The second part, which is iteration dependent, is need to be repeated every iterationin in two steps: first, the ${Q}$ scalar multiplications by the $U$ vectors in \eqref{eq: solution_Xq}; second, a multiplication of the $U\times U$ matrix by the $U\times 1$ vector and the ${Q}$ scalar multiplications by the $U\times 1$ vectors in \eqref{eq: solution X0}. These steps require $\frac{1}{2}QU$, $U^2+\frac{1}{2}QU$ complex multiplications, respectively. Combining this result with Theorem \ref{iteration complexity}, we conclude that the total computational cost to attain an $\epsilon$-optimal solution is $\frac{1}{3}U^3+\frac{1}{2}BU^2+BU+K(U^2+QU)$, where the number of iterations $K=\frac{1}{C\epsilon}\Big(L_{\rho}(\{\mathbf{x}^{1}_q\}_{q=1}^{Q}, \mathbf{x}_0^{1}, \mathbf{y}^{1}) - \big(\frac{1}{2}\Vert\mathbf{r}-\mathbf{H}\mathbf{x}_0^* \Vert_2 ^{2} - \sum_{q=1}^{Q}\frac{\alpha_q}{2} \Vert\mathbf{x}_q^* \Vert_2 ^{2}\big)\!\Big)$. Since $K \ll B$ for massive MIMO detection, the computational complexity of the PS-ADMM is comparable to that of the linear detector.

\footnotetext{The computational complexity is measured by the number of complex-valued multiplications for $K$ iterations. A complex-valued multiplication is assumed to required four real-valued multiplications.}

\section{Simulation results}\label{sec:Simulation results}

In this section, numerical results are presented to show the effectiveness of the proposed PS-ADMM detector.
Specifically, in Section \ref{simulation-result-performance} we demonstrate the BER performance of the PS-ADMM detector compare with everal existing detectors.
In Section \ref{simulation-parameter-setting}, we focus on analyzing the impact of parameters $\rho$, $\{\alpha_q\}_{q=1}^{Q}$ and $K$ on performance of the PS-ADMM detector.

Throughout this section, we show simulation results for uncoded signal detection with the i.i.d rayleigh fading channel in different $B \times U \;(B \ge U)$ massive MU-MIMO systems. The modulation schemes of QPSK, 16-QAM and 64-QAM are employed. We assume perfect knowledge of the channel state information is exactly known at the receiver.
For a fair comparison, all the algorithms are implemented in Mathworks Matlab 2019a/Windows 7 environment on a computer with 3.7GHz Intel i3-6100$\times$2 CPU and 16GB RAM.

\begin{figure*}[htpb]
\subfigure[$B=128,U=16$ for QPSK.]{
    \begin{minipage}{8.5cm}
    \centering
        \includegraphics[width=3.5in,height=2.7in]{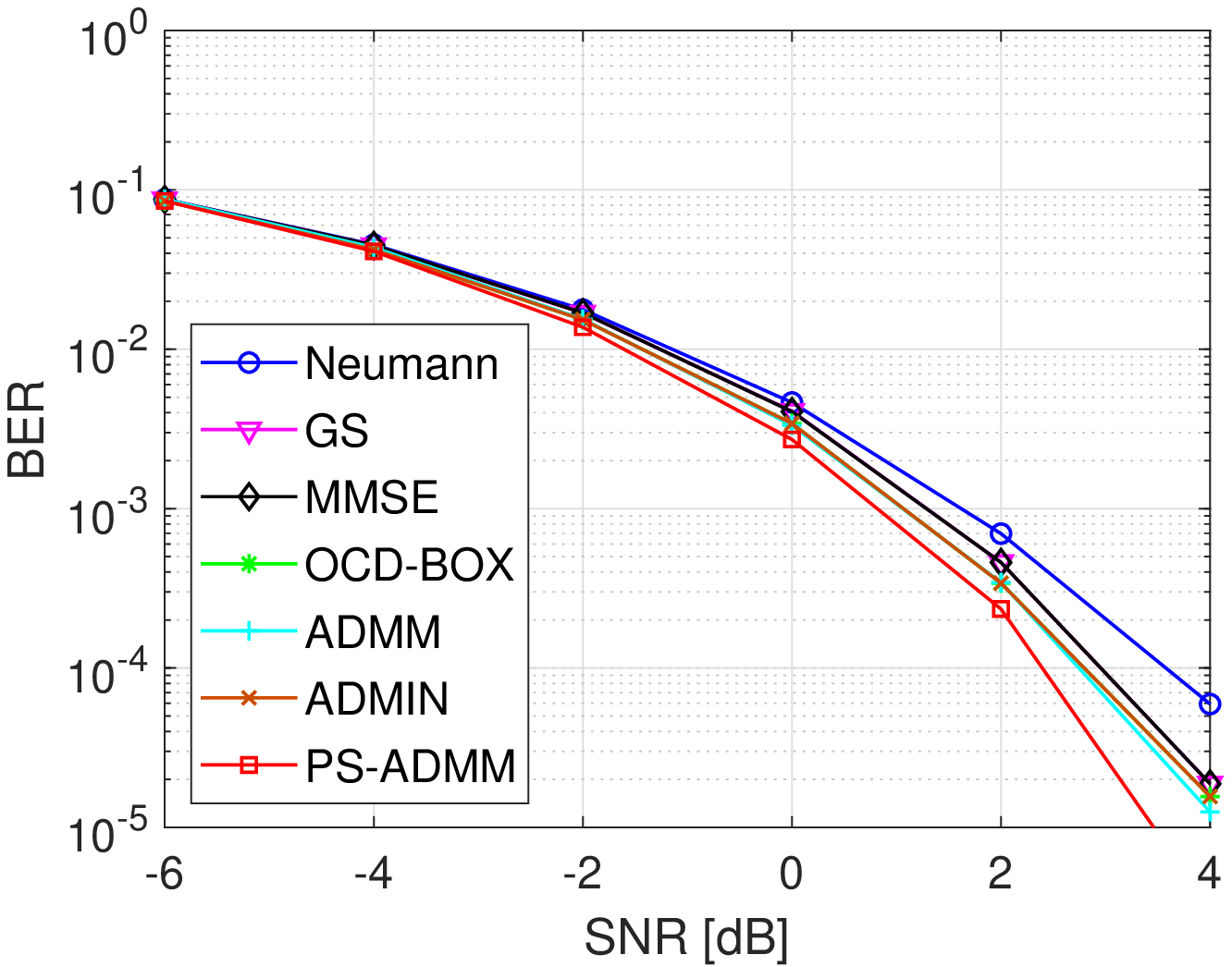}
            \label{ber-all-QPSK-128x16}
    \end{minipage}
    }
 \subfigure[$B=128,U=32$ for QPSK.]{
    \begin{minipage}{8.5cm}
    \centering
        \includegraphics[width=3.5in,height=2.7in]{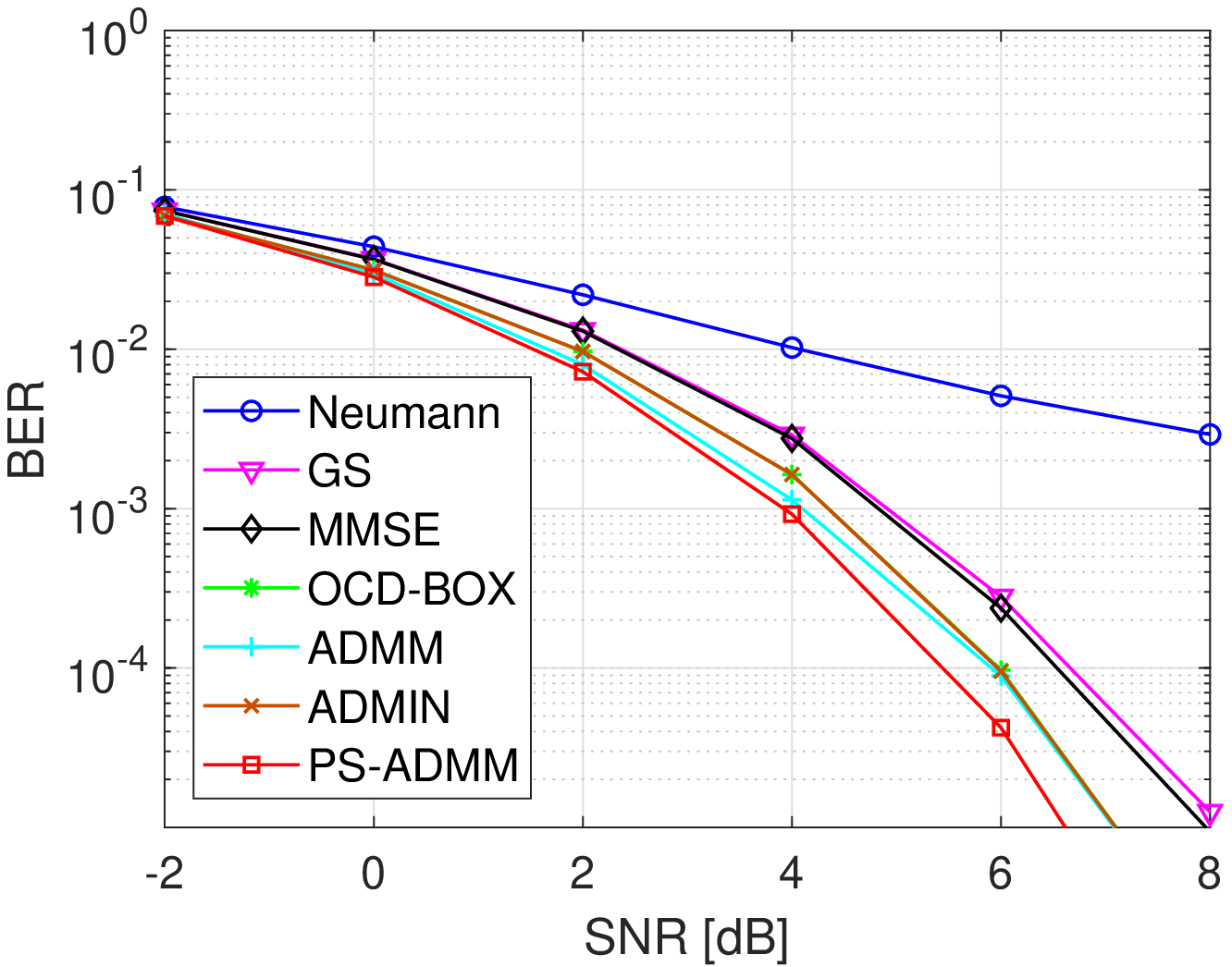}
            \label{ber-all-QPSK-128x32}
    \end{minipage}
    }

\subfigure[$B=128,U=64$ for QPSK.]{
    \begin{minipage}{8.5cm}
    \centering
        \includegraphics[width=3.5in,height=2.7in]{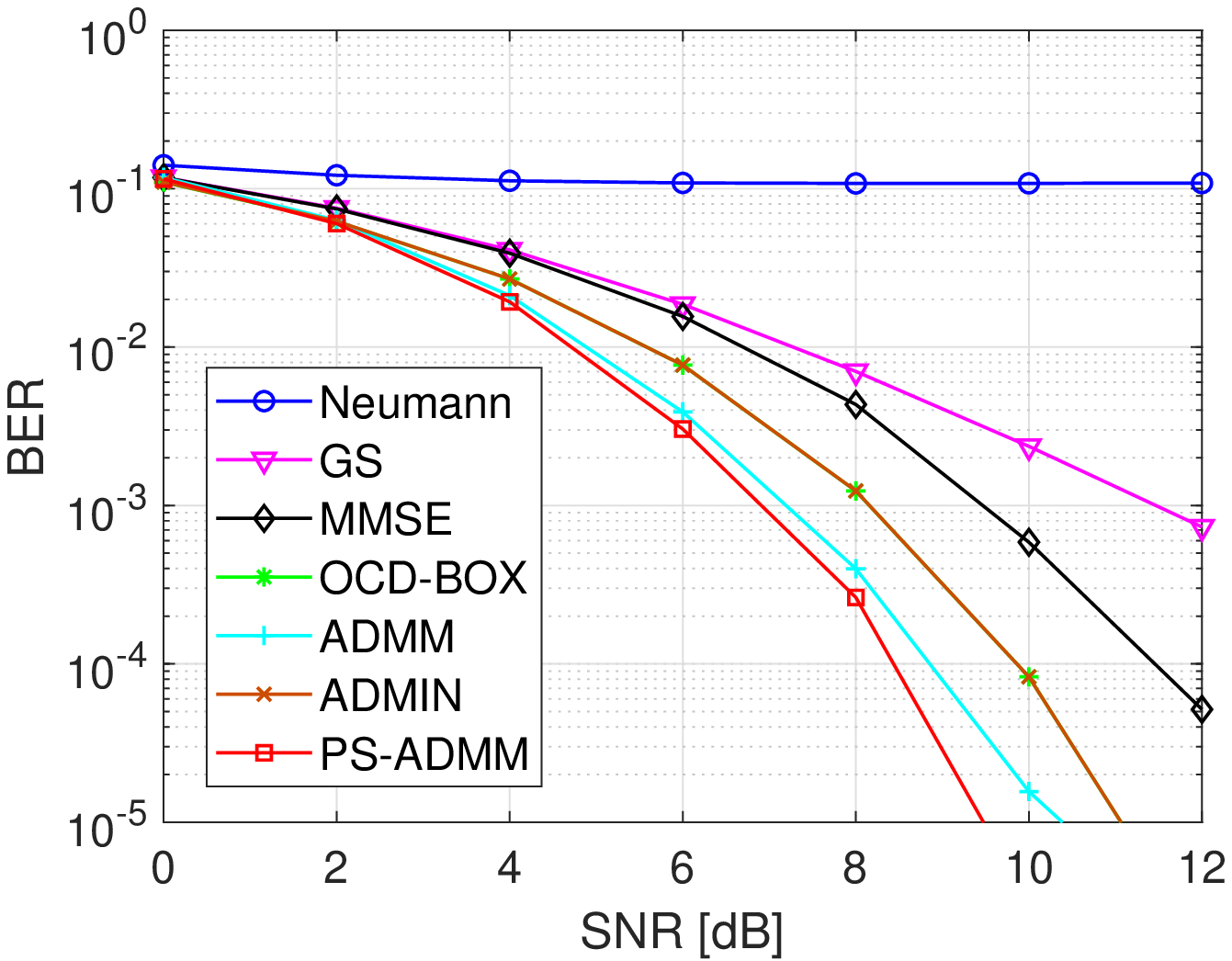}
            \label{ber-all-QPSK-128x64}
    \end{minipage}
    }
 \subfigure[$B=128,U=128$ for QPSK.]{
    \begin{minipage}{8.5cm}
    \centering
        \includegraphics[width=3.5in,height=2.7in]{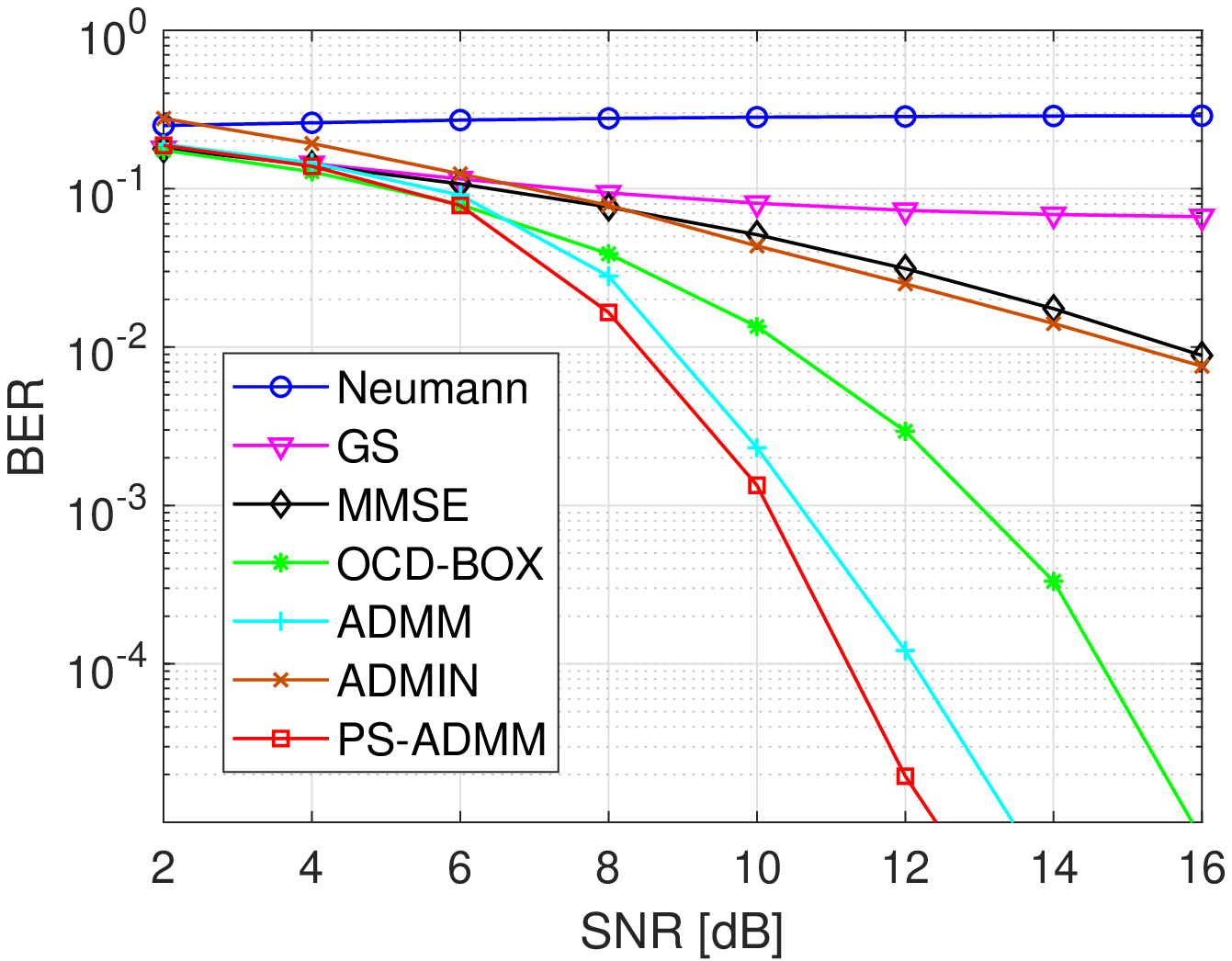}
            \label{ber-all-QPSK-128x128}
    \end{minipage}
    }
 \subfigure[$B=128,U=16$ for 16-QAM.]{
    \begin{minipage}{8.5cm}
    \centering
        \includegraphics[width=3.5in,height=2.7in]{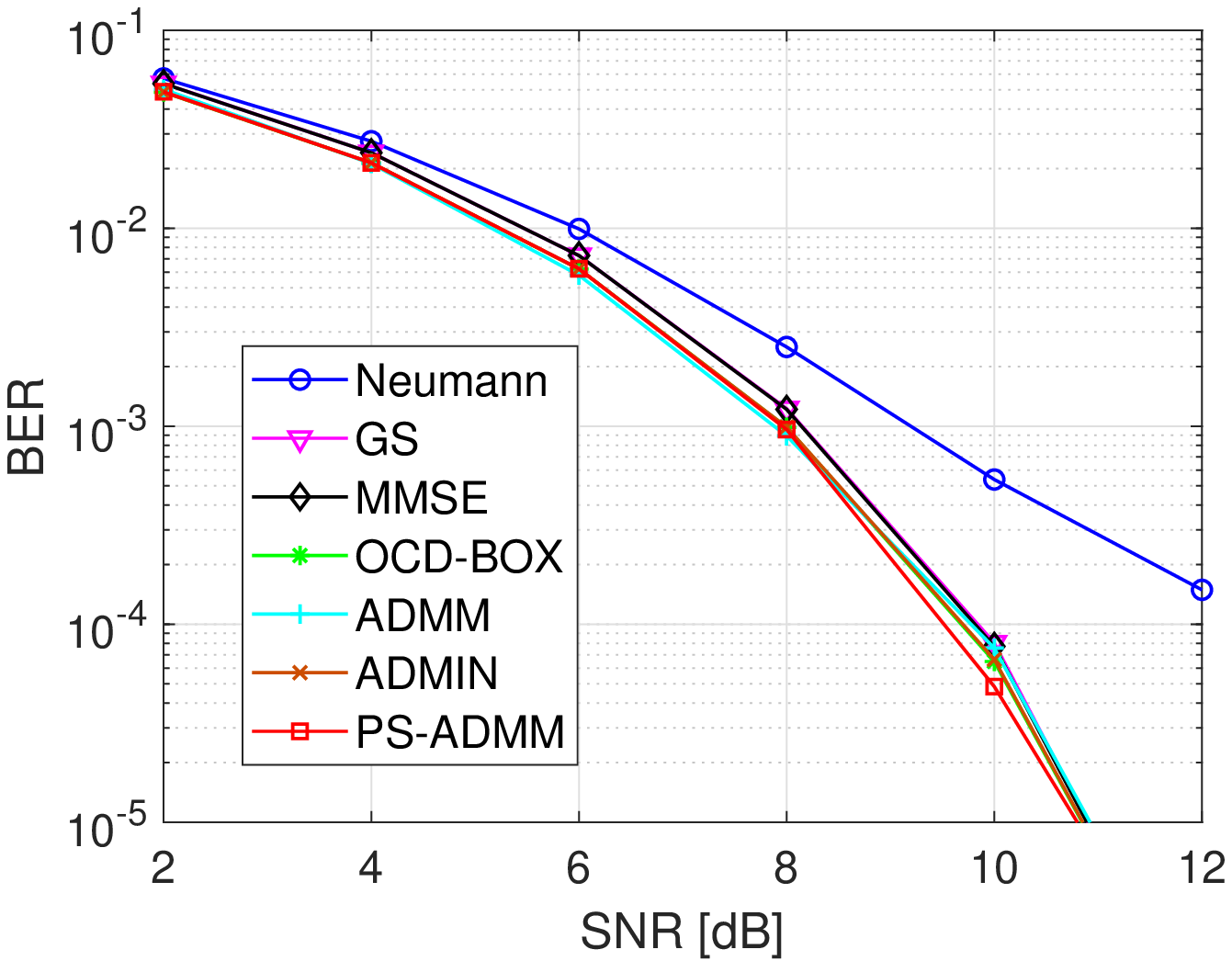}
            \label{ber-all-16qam-128x16}
    \end{minipage}
    }
 \subfigure[$B=128,U=32$ for 16-QAM.]{
    \begin{minipage}{8.5cm}
    \centering
        \includegraphics[width=3.5in,height=2.7in]{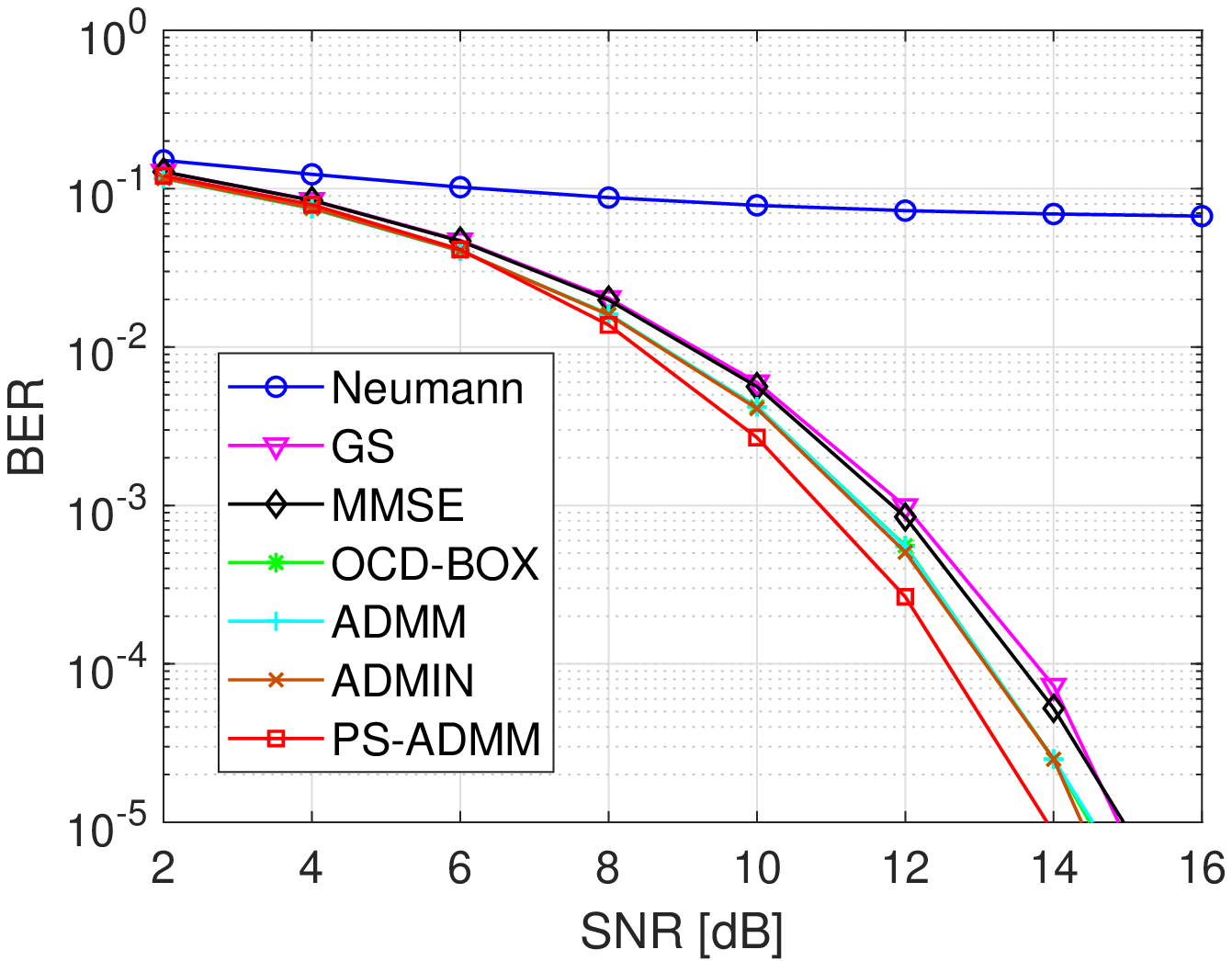}
            \label{ber-all-16qam-128x32}
    \end{minipage}
    }

\end{figure*}
\addtocounter{figure}{-1} 
\begin{figure*}
%\ContinuedFloat
\addtocounter{figure}{1} 

 \subfigure[$B=128,U=64$ for 16-QAM.]{
    \begin{minipage}{8.5cm}
    \centering
        \includegraphics[width=3.5in,height=2.7in]{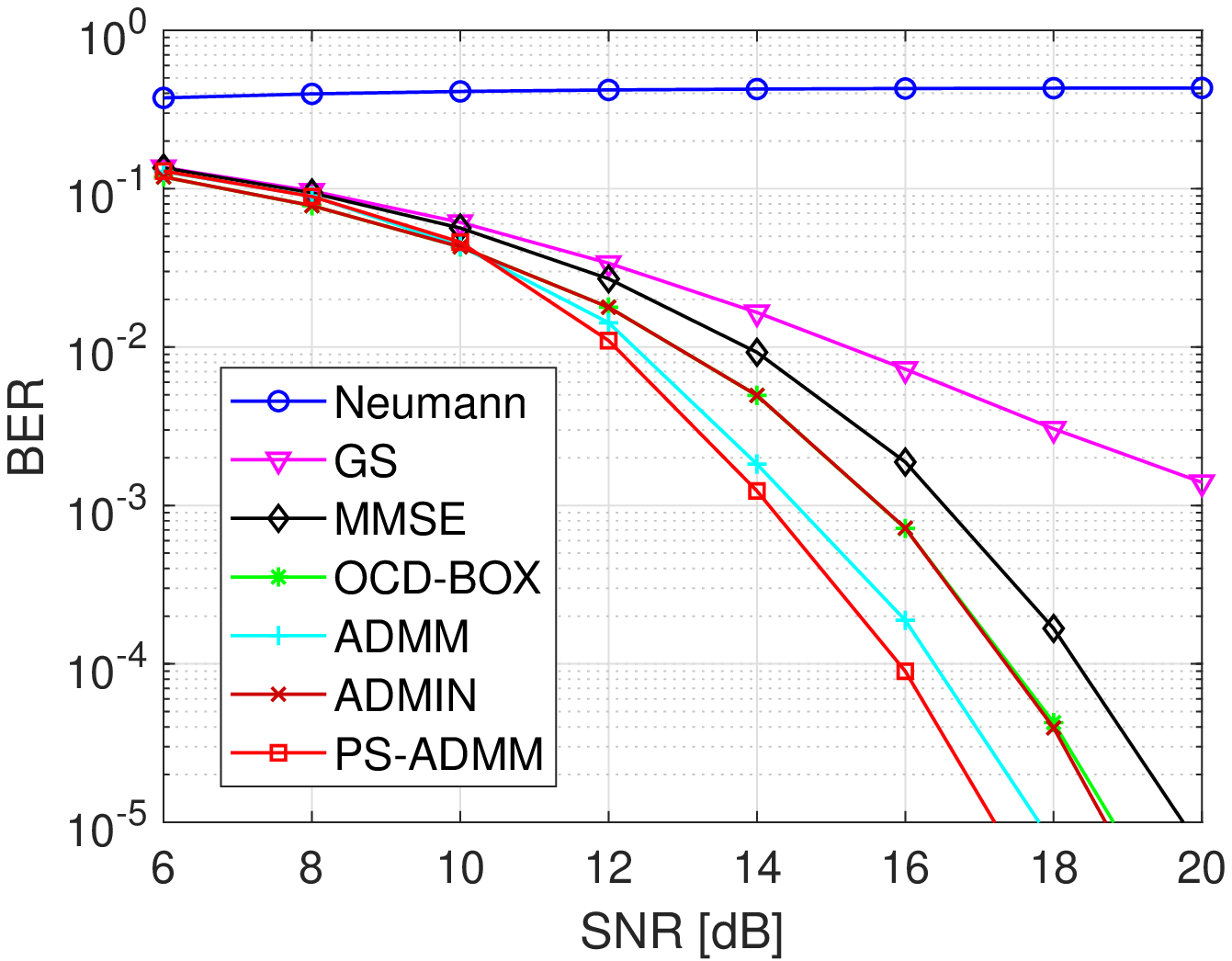}
            \label{ber-all-16qam-128x64}
    \end{minipage}
    }
 \subfigure[$B=128,U=128$ for 16-QAM.]{
    \begin{minipage}{8.5cm}
    \centering
        \includegraphics[width=3.5in,height=2.7in]{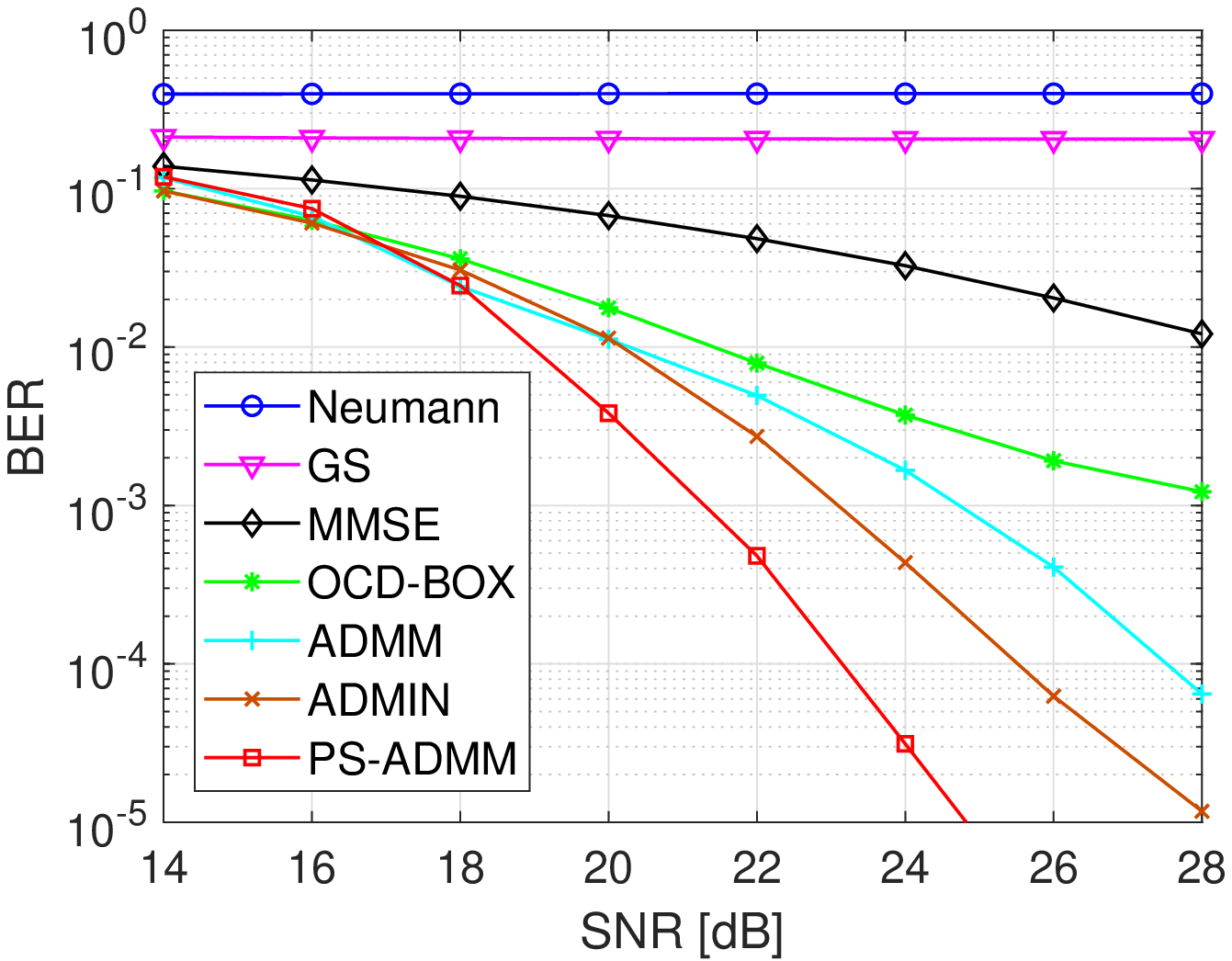}
            \label{ber-all-16qam-128x128}
    \end{minipage}
    }
\subfigure[$B=128,U=16$ for 64-QAM.]{
    \begin{minipage}{8.5cm}
    \centering
        \includegraphics[width=3.5in,height=2.7in]{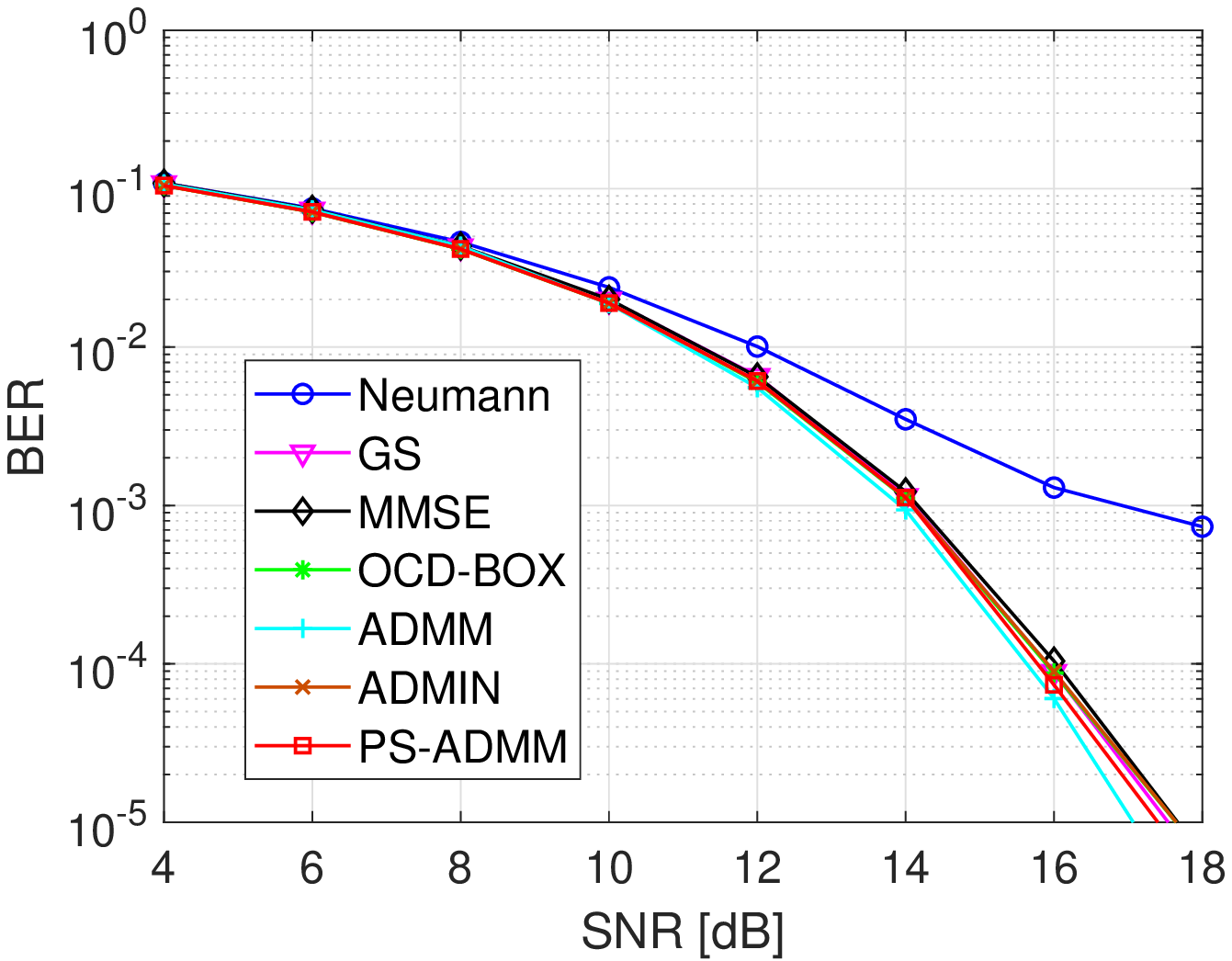}
            \label{ber-all-64qam-128x16}
    \end{minipage}
    }
\subfigure[$B=128,U=32$ for 64-QAM.]{
    \begin{minipage}{8.5cm}
    \centering
        \includegraphics[width=3.5in,height=2.7in]{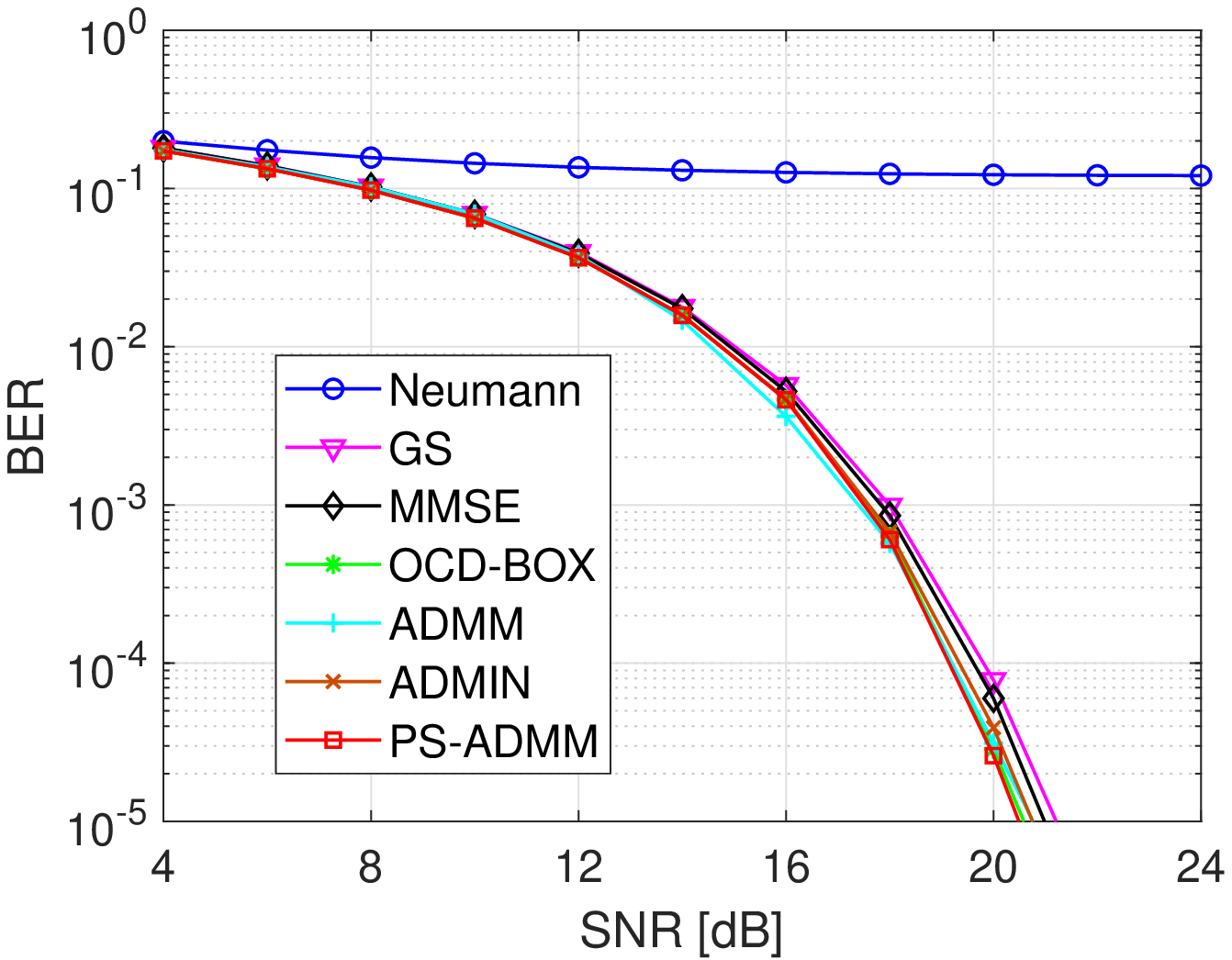}
            \label{ber-all-64qam-128x32}
    \end{minipage}
    }
\subfigure[$B=128,U=64$ for 64-QAM.]{
    \begin{minipage}{8.5cm}
    \centering
        \includegraphics[width=3.5in,height=2.7in]{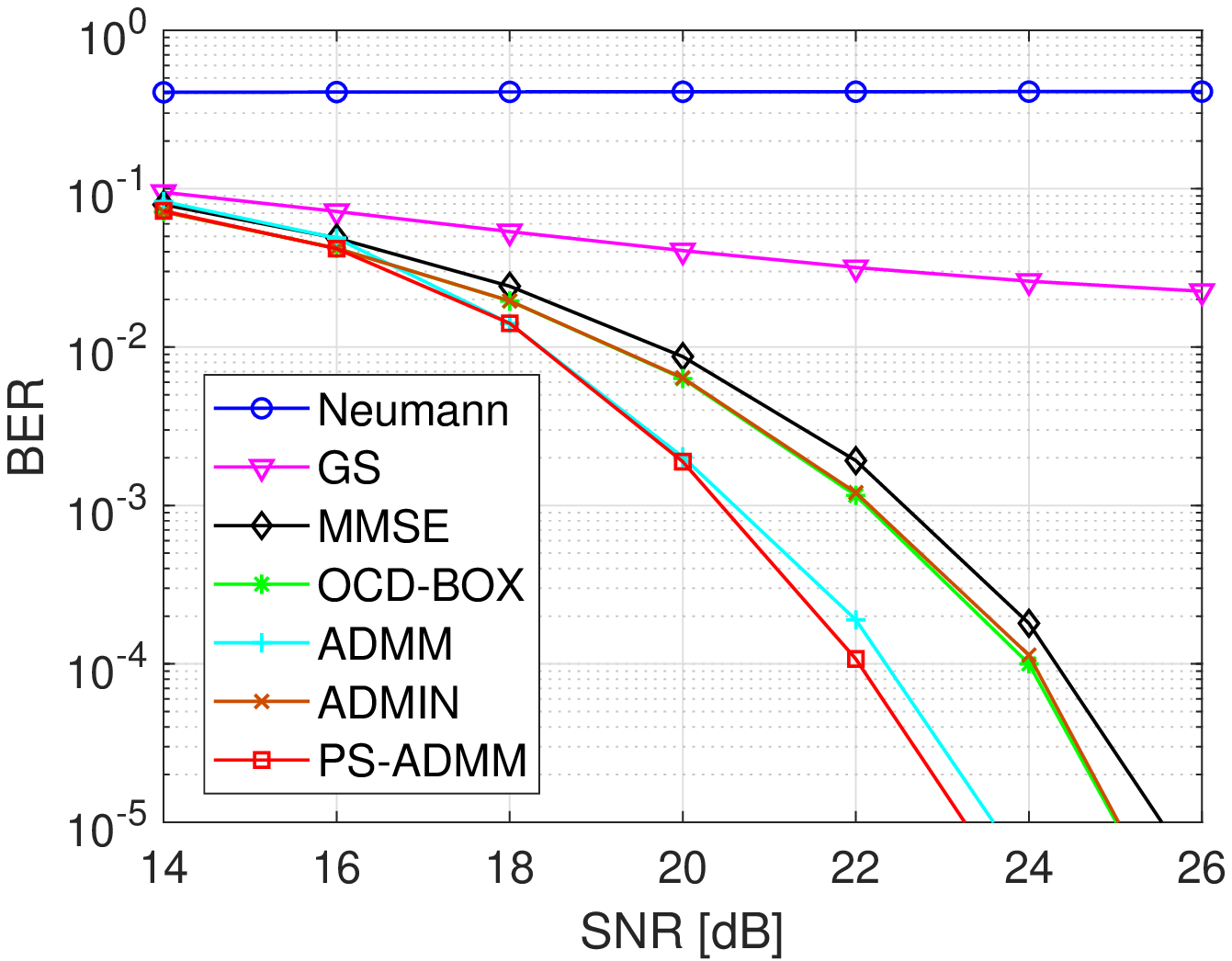}
            \label{ber-all-64qam-128x64}
    \end{minipage}
    }
  \subfigure[$B=128,U=128$ for 64-QAM.]{
    \begin{minipage}{8.5cm}
    \centering
        \includegraphics[width=3.5in,height=2.7in]{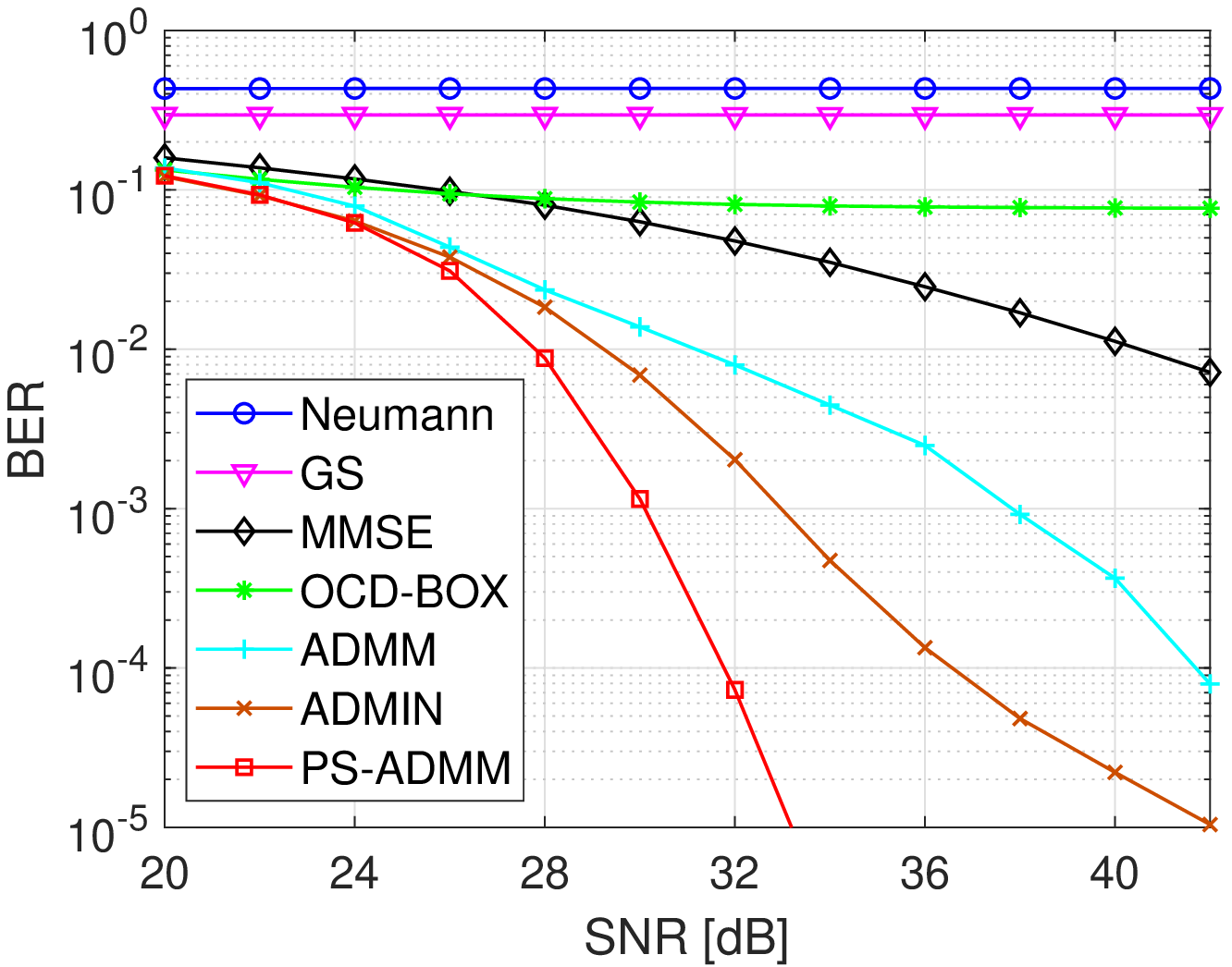}
            \label{ber-all-64qam-128x128}
    \end{minipage}
    }
    \centering
    \caption{Comparisons of BER performance using various massive MIMO detectors.}
    \label{ber_allalgo}
 \end{figure*}

\begin{figure}[tp]
   \subfigure[BER performance vs. $\rho$.]{
    \begin{minipage}{9cm}
    \centering
        \includegraphics[width=3.5in,height=2.7in]{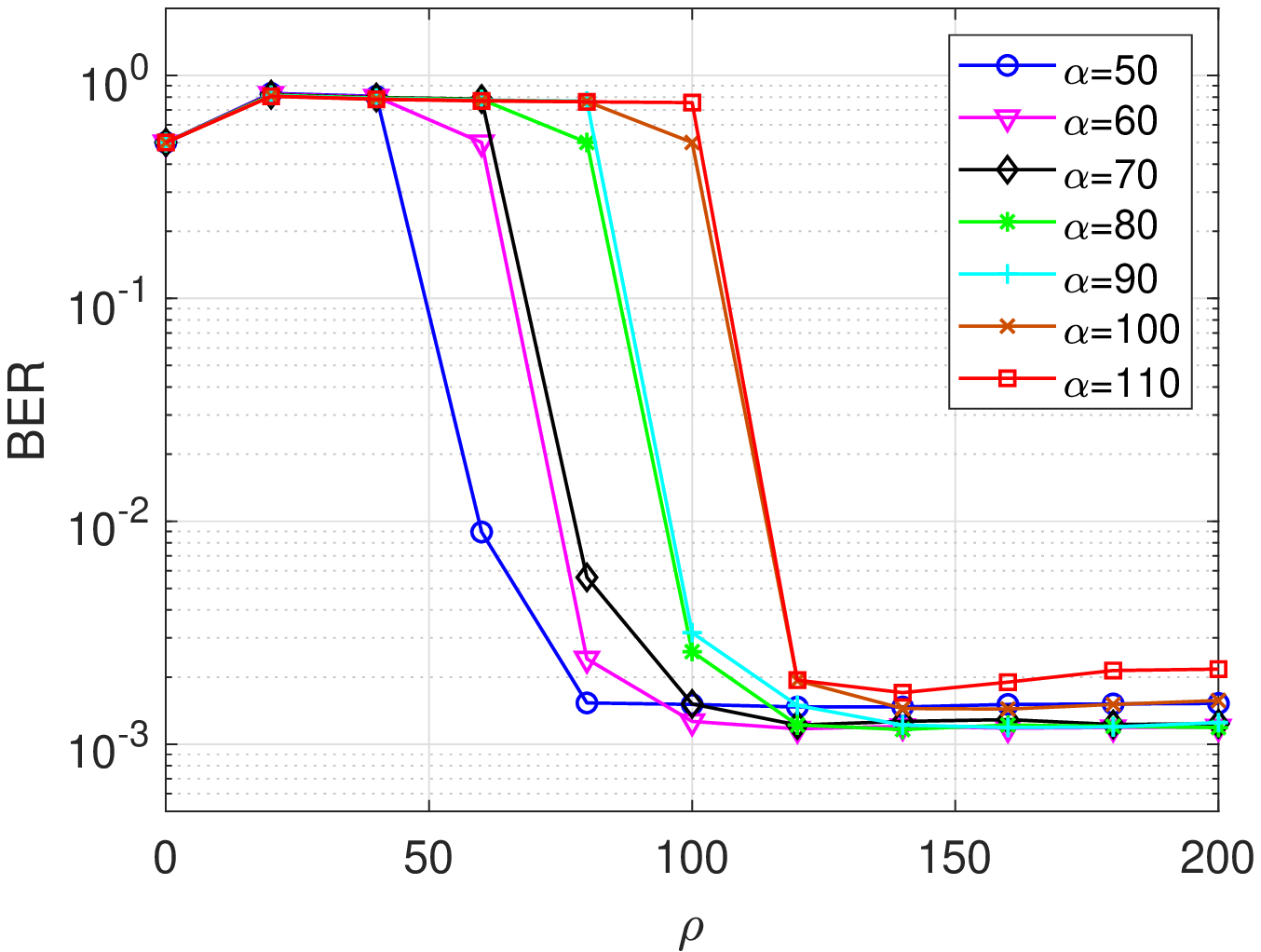}
            \label{ber_rho_qpsk}
    \end{minipage}
    }
   % ~~~~~~~~~~~~~
   \subfigure[BER performance vs. $\alpha$.]{
    \begin{minipage}{9cm}
    \centering
        \includegraphics[width=3.5in,height=2.7in]{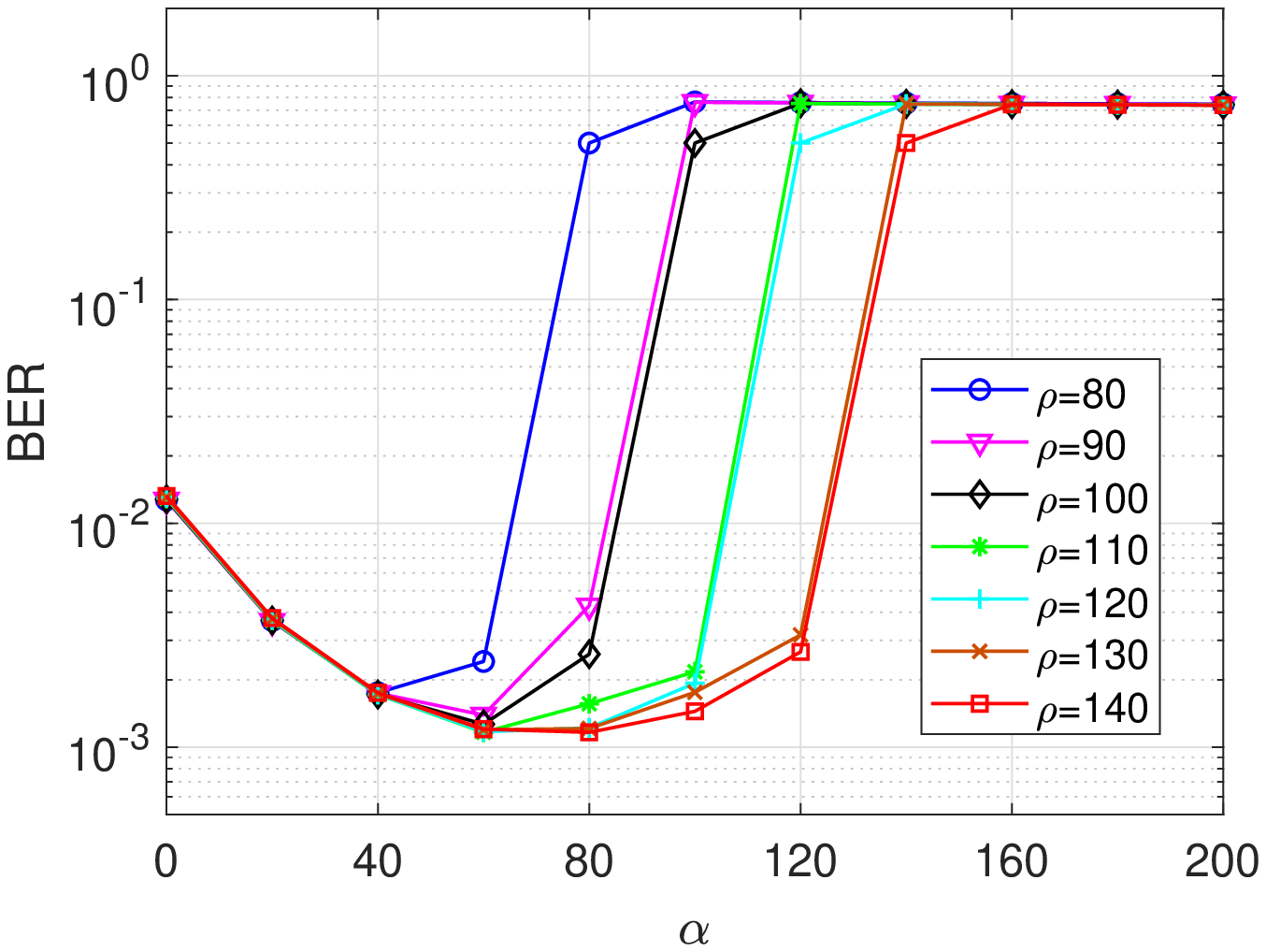}
            \label{ber_alpha_qpsk}
    \end{minipage}
    }
 \centering
 \caption{The impact of $\rho$ and $\alpha$ on BER performance of the PS-ADMM detector for $B=128,U=128$, SNR=10dB, QPSK modulation.}
 \label{impact_BER}
\end{figure}

 \begin{figure}[tp]

    \subfigure[Convergence characteristic comparison when $\rho=300$.]
    {
    \begin{minipage}{9cm}
    \centering
        \includegraphics[width=3.5in,height=2.7in]{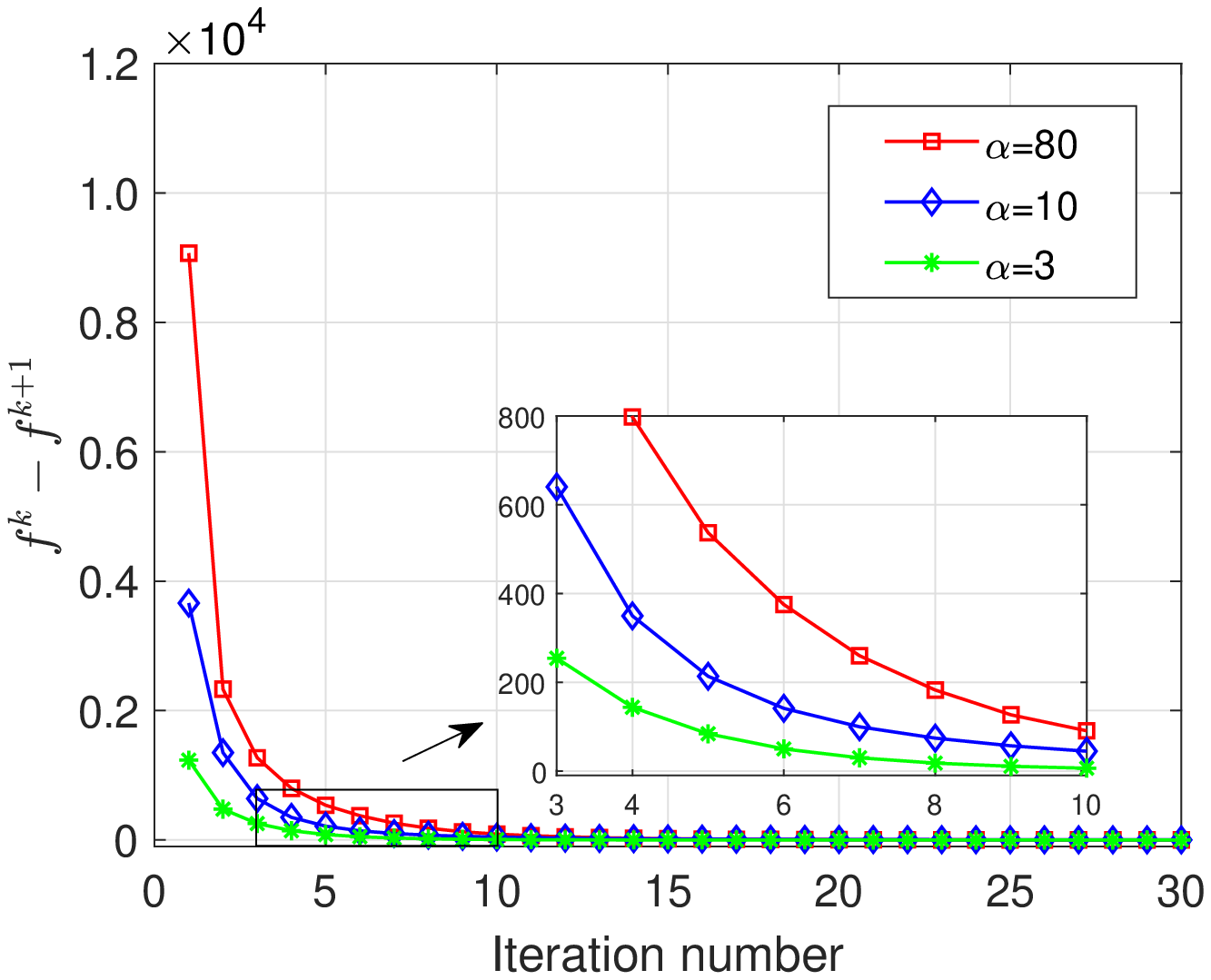}
            \label{conv_alpha_qpsk}
    \end{minipage}
    }
        \subfigure[Convergence characteristic comparison when $\alpha=80$.]
    {
    \begin{minipage}{9cm}
    \centering
        \includegraphics[width=3.5in,height=2.7in]{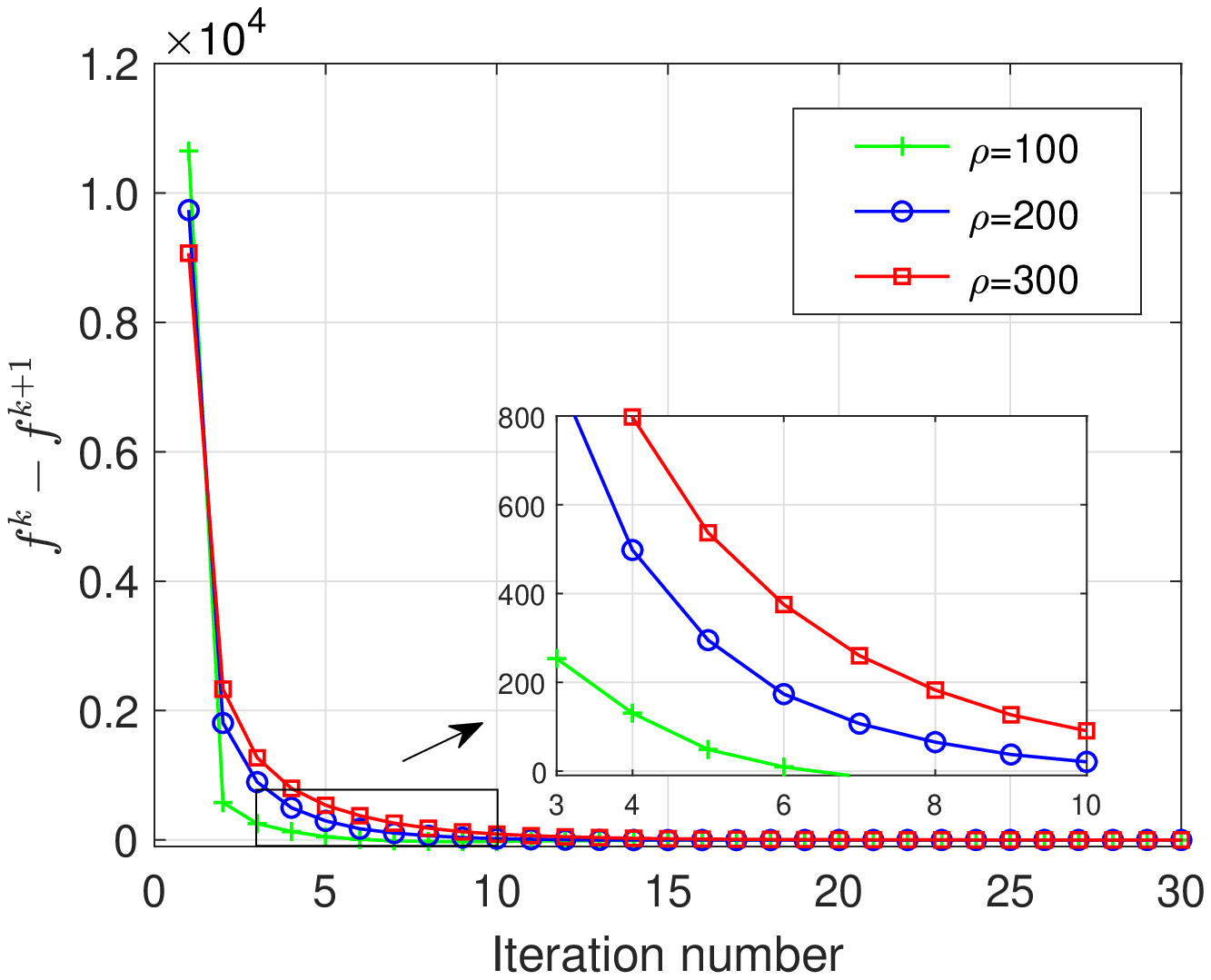}
            \label{conv_rho_qpsk}
    \end{minipage}
    }
 \centering
 \caption{The impact of $\rho$ and $\alpha$ on Convergence performance of the PS-ADMM detector for $B=128,U=128$, SNR=10dB, QPSK modulation.}
 \label{impact_conv}
\end{figure}

\subsection{BER performance}\label{simulation-result-performance}
In this subsection, the BER performance of the proposed PS-ADMM detector was evaluated and compared with some conventional and state-of-the-art MIMO detectors by numerical simulations. The considered detectors are 
classical MMSE detector, Neumann detector \cite{wu2014large}, GS detector \cite{wu2016efficient}, OCD-BOX detector \cite{wu2016high}, and two ADMM-based detectors ADMM and ADMIN in \cite{7526551} and \cite{shahabuddin2017admm} respectively. The penalty parameters $\rho$ and $\{\alpha_q\}_{q=1}^{Q}$ are chosen the best value in different simulation scenarios respectively. The termination criteria is set as the $K=30$ is reached. The points plotted in all BER curves are based on generating at least 1000 Monte-Carlo trials.

 Fig.\,\ref{ber_allalgo} shows BER performance of considered detectors for QPSK, 16-QAM and 64-QAM modulation with different number of transmit antennas massive MIMO systems. 
 In Fig.\,\ref{ber_allalgo}, one can see that BER curves of all detectors have a similar changing trend at low SNRs and continues to drop in a waterfall manner in relatively high SNR regions. We observe that PS-ADMM detector achieves the best BER performance.
 Observing Fig.\,\ref{ber-all-QPSK-128x16}-\ref{ber-all-QPSK-128x64}, \ref{ber-all-16qam-128x16}-\ref{ber-all-16qam-128x64}, and \ref{ber-all-64qam-128x16}-\ref{ber-all-64qam-128x64}, we can find that all detectos are displays comparable BER performance when the BS-to-user-antenna ratio is more than two, only the approximate matrix inversion algorithms such as Neumann and GS suffer from a performance loss.
 In Fig.\,\ref{ber-all-QPSK-128x128},\,\ref{ber-all-16qam-128x128} and \,\ref{ber-all-64qam-128x128}, for the more challenging square 128 $\times$ 128 massive MIMO systems, we see that the PS-ADMM exhibits excellent performance and outperforms that of all the other detectors.
 It can be concluded that the proposed PS-ADMM detector achieves excellent BER performance for different massive MIMO systems, especially when the BS-to-user-antenna ratio is close to one.

 \subsection{Choice of Parameters} \label{simulation-parameter-setting}

We show several simulations for the significantly affect on detecting performance by the choosing of $\rho$ and $\alpha$ for QPSK modulationin system in Fig.\,\ref{impact_BER} and Fig.\,\ref{impact_conv}.
From them, we can choose proper parameters $\rho$ and $\alpha$ to achieve the lower BER and faster convergence.

In Fig.\,\ref{impact_BER}, the BER performance as a function of the penalty parameter $\rho$ and $\alpha$ are presented.
We make the following observations from them, first, all of the curves of the BER performance display a similar changing trend as parameters $\rho$ and $\alpha$ change. Second, BER performance is sensitive to parameters $\alpha$ and be slightly affected by $\rho$ when $\rho$ is large enough, a usable BER can be abttained only when $\rho>\alpha$ is satisfied. Third, a set of parameters $\rho$ and $\alpha$ can achieves the lowest BER performance, as we increases or decreases the optimal $\alpha$, the corresponding BER increases. If $\alpha$ even decreases to zero, the PS-ADMM detector becomes the conventional ADMM detector with box constraint. 

In Fig.\,\ref{impact_conv}, we study the effects of the penalty parameters $\rho$ and $\alpha$ on convergence characteristic. To illustrate the convergence curve clearly, we choose the linear terms as the vertical axis, where $f^{k}$ denotes the objective in \eqref{eq:PS_ADMM} at $k$ iteration.
From it, we can observe that, first, all of the curves show a similar tendency that convergence performance change over the iteration numbers with different parameters $\rho$ and $\alpha$.
Second, parameters $\rho$ and $\alpha$ can affect the convergence rate of the PS-ADMM algorithm. The smaller the penalty parameter $\rho$ and $\alpha$, the faster the optimization problem to converge.
Third, the number of iterations $K$ does not significantly affect the convergence characteristic of objective when $K>20$. The PS-ADMM algorithm is performed within a few tens of iterations to converge to a modest accuracy solutions, which is sufficient for the kinds of large-scale problems such as massive MIMO detection.

\section{Conclusion}
\label{sec:Conclusion}

In this paper, to the best of our knowledge, we proposed the first signal detection that utilizes penalty-sharing ADMM approach. Our novel detector called PS-ADMM is suitable for massive MIMO systems with high-order QAM modulations. The key idea is to add a set of penalty terms that converted from high-dimensional signals to the objective function of formulated sharing ML optimization problem with bound relaxation constraints, and then the sharing-ADMM technique be implemented in parallel to solve the nonconvex problem. We prove that the proposed PS-ADMM approach is theoretically-guaranteed convergent if proper parameters are chosen. Compared with several existing massive MIMO detectors, the PS-ADMM detector demonstrates excellent BER performance while cheap computational complexity providing, especially when the BS-to-user-antenna ratio is close to one. In addition, though our focus is on massive MIMO detection, the proposed PS-ADMM method would also be a natural fit for high-dimensional large-scale nonconvex optimization problem in some practical applications.

\appendices

\section{Proof of Lemmas \ref{lemma:z1}--\ref{lemma:L_bounded1}}\label{lemma1-3}
Our convergence analysis consists of a series of lemmas, we first show that the size of the successive difference of the dual variables can be bounded above by that of the auxiliary variables.
\begin{lemma}\label{lemma:z1}
For Algorithm 1 , the following is true
\begin{align}\label{eq:y_difference1}
\|\mathbf{y}^{k+1}-\mathbf{y}^k\|_2^2 \le \lambda_{\rm max}^2(\mathbf{H}^H\mathbf{H})\|\mathbf{x}_0^{k+1}-\mathbf{x}_0^k\|_2^2.
\end{align}
\end{lemma}
\begin{proof}
From the $\mathbf{x}_0$ in update step \eqref{eq:x0_update}, we have the following optimality condition
\begin{align}\label{eq:x_nabla_expression}
\nabla_{\mathbf{x}_0} \ell\left(\mathbf{x}_0^{k+1}\right)+\mathbf{y}^k+\rho (\mathbf{x}_0^{k+1}-\sum_{q=1}^{Q}2^{q-1} \mathbf{x}_q^{k+1})=0.
\end{align}
Combined with the dual variable update step \eqref{eq:y_update} we obtain
\begin{equation}\label{eq:z_nabla_expression}
\mathbf{y}^{k+1}=-\nabla_{\mathbf{x}_0} \ell\left(\mathbf{x}_0^{k+1}\right).
\end{equation}
Since we known
\begin{equation}
\nabla_{\mathbf{x}_0}^2 \ell\left(\mathbf{x}_0^{k+1}\right) =\mathbf{H}^{H} \mathbf{H} \preceq \lambda_{\rm max}(\mathbf{H}^{H} \mathbf{H})\mathbf{I},\nonumber
\end{equation}
i.e.,there exists such Lipschitz continuous that
\begin{equation}\label{eq:boundy_expression}
\begin{split}
&\|\nabla_{\mathbf{x}_0} \ell\left(\mathbf{x}_0^{k+1}\right)-\nabla_{\mathbf{x}_0} \ell\left(\mathbf{x}_0^{k}\right)\|_2^2 \\
 \le &\lambda_{\rm max}^2(\mathbf{H}^{H} \mathbf{H}) \|\mathbf{x}_0^{k+1}-\mathbf{x}_0^{k}\|_2^2.
 \end{split}
\end{equation}
Combining \eqref{eq:z_nabla_expression} with \eqref{eq:boundy_expression}, we can obtain
\begin{align*}
\begin{split}
&\|\mathbf{y}^{k+1}-\mathbf{y}^k\|_2^2 \\
=&\|\nabla_{\mathbf{x}_0} \ell\left(\mathbf{x}_0^{k+1}\right)-\nabla_{\mathbf{x}_0} \ell\left(\mathbf{x}_0^{k}\right)\|_2^2 \\
\le & \lambda_{\rm max}^2(\mathbf{H}^{H} \mathbf{H}) \|\mathbf{x}_0^{k+1}-\mathbf{x}_0^{k}\|_2^2.
\end{split}
\end{align*}
This completes the proof.
$\hfill\blacksquare$
\end{proof}
\begin{lemma}\label{lemma:L_difference}
Assume parameters $\{\alpha_q\}_{q=1}^{Q}$ and $\rho$ satisfy $4^{q-1}\rho>\alpha_q$ and $\rho > \sqrt{2} \lambda_{\rm max}(\mathbf{H}^H\mathbf{H})$. Then for Algorithm 1, we have the following
\begin{align}\label{eq:L_difference}
&L_{\rho}\left (\{\mathbf{x}_q^{k+1}\}_{q=1}^{Q}, \mathbf{x}_0^{k+1}, \mathbf{y}^{k+1}\right)-L_{\rho}\left (\{\mathbf{x}_q^{k}\}_{q=1}^{Q}, \mathbf{x}_0^{k}, \mathbf{y}^{k} \right)\nonumber\\
&\le \sum_{q=1}^{Q}-\frac{\gamma_q(\rho)}{2}\|\mathbf{x}_q^{k+1}-\mathbf{x}_{q}^{k}\|_2^2 \nonumber\\
&-\Big(\frac{\gamma(\rho)}{2}-\frac{\lambda_{\rm max}^2(\mathbf{H}^{H} \mathbf{H})}{\rho}\Big)\|\mathbf{x}_0^{k+1}-\mathbf{x}_0^k\|_2^2,
\end{align}
\end{lemma}
where $\gamma_q(\rho)$ and $\gamma(\rho)$ are strongly convex moduli of $ L_{\rho}(\{\mathbf{x}_q\}_{q=1}^{Q}, \mathbf{x}_0, \mathbf{y})$ w.r.t.\ each $\mathbf{x}_q$ and $\mathbf{x}_0$, respectively.
\begin{proof}
We first split the successive difference of the augmented Lagrangian by
\begin{align}\label{eq:successive_L1}
&L_{\rho}(\{\mathbf{x}_q^{k+1}\}, \mathbf{x}_0^{k+1}, \mathbf{y}^{k+1})-L_{\rho}(\{\mathbf{x}^{k}_q\}, \mathbf{x}_0^{k}, \mathbf{y}^{k})\nonumber\\
&=\left(L_{\rho}(\{\mathbf{x}_q^{k+1}\}, \mathbf{x}_0^{k+1}, \mathbf{y}^{k+1})-L_{\rho}(\{\mathbf{x}_q^{k+1}\}, \mathbf{x}_0^{k+1}, \mathbf{y}^{k})\right)\nonumber\\
&\quad+\left(L_{\rho}(\{\mathbf{x}^{k+1}_q\}, \mathbf{x}_0^{k+1}, \mathbf{y}^{k})-L_{\rho}(\{\mathbf{x}^{k}_q\}, \mathbf{x}_0^{k}, \mathbf{y}^{k})\right),\nonumber\\
&\hspace{6cm}q=1,2,\cdots,Q
\end{align}
The first term in \eqref{eq:successive_L1} can be bounded by
\begin{align}
\begin{split}\label{eq:extra}
&L_{\rho}(\{\mathbf{x}_q^{k+1}\}_{q=1}^{Q}, \mathbf{x}_0^{k+1}, \mathbf{y}^{k+1})-L_{\rho}(\{\mathbf{x}_q^{k+1}\}_{q=1}^{Q}, \mathbf{x}_0^{k+1}, \mathbf{y}^{k})\\
&={\rm Re}\big\langle \mathbf{x}_0^{k+1}-\sum_{q=1}^{Q}2^{q-1} \mathbf{x}_q^{k+1}, \; \mathbf{y}^{k+1}\big\rangle \\
&\quad-{\rm Re}\big\langle \mathbf{x}_0^{k+1}-\sum_{q=1}^{Q}2^{q-1} \mathbf{x}_q^{k+1}, \; \mathbf{y}^{k}\big\rangle \\
&={\rm Re}\big\langle \mathbf{x}_0^{k+1}-\sum_{q=1}^{Q}2^{q-1} \mathbf{x}_q^{k+1},\; \mathbf{y}^{k+1}-\mathbf{y}^{k}\big\rangle \\
&\stackrel{\rm (a)}
=\frac{1}{\rho}\|\mathbf{y}^{k+1}-\mathbf{y}^{k}\|_2^2\\
&\stackrel{\rm (b)}
\le\frac{\lambda_{\rm max}^2(\mathbf{H}^{H} \mathbf{H})}{\rho}\|\mathbf{x}_0^{k+1}-\mathbf{x}_0^{k}\|_2^2,
\end{split}
\end{align}
in \eqref{eq:extra}, the equality ``$\stackrel{\rm (a)}{=}$'' and ``$\stackrel{\rm (b)}{\le}$'' hold since we have use \eqref{eq:y_update} and \eqref{eq:y_difference1} respectively.
The second term in \eqref{eq:successive_L1} can be bounded by
\begin{align}
\begin{split}\label{eq:extra1}
&L_{\rho}(\{\mathbf{x}_q^{k+1}\}_{q=1}^{Q}, \mathbf{x}_0^{k+1}, \mathbf{y}^{k})-L_{\rho}(\{\mathbf{x}_q^{k}\}_{q=1}^{Q}, \mathbf{x}_0^{k}, \mathbf{y}^{k})\\
&=L_{\rho}(\{\mathbf{x}_q^{k+1}\}_{q=1}^{Q}, \mathbf{x}_0^{k}, \mathbf{y}^{k})-L_{\rho}(\{\mathbf{x}_q^{k}\}_{q=1}^{Q}, \mathbf{x}_0^{k}, \mathbf{y}^{k})\\
&+L_{\rho}(\{\mathbf{x}_q^{k+1}\}_{q=1}^{Q}, \mathbf{x}_0^{k+1}, \mathbf{y}^{k})-L_{\rho}(\{\mathbf{x}_q^{k+1}\}_{q=1}^{Q}, \mathbf{x}_0^{k}, \mathbf{y}^{k})\\
&\stackrel{\rm (a)}\le \sum_{q=1}^{Q}\Bigl({\rm Re}\left\langle\nabla_{\mathbf{x}_q}  L_{\rho}(\mathbf{x}_q^{k+1},\!\mathbf{x}_1^{k+1}\!\!\!\!,\!\cdots\!,\mathbf{x}_{q-1}^{k+1}\!, \mathbf{x}_{q+1}^k\!,\!\cdots\!,\!\mathbf{x}_Q^k,\!\mathbf{x}_0^{k},\! \mathbf{y}^{k}), \right.\\
&\hspace{1cm}\left.\mathbf{x}_q^{k+1}-\mathbf{x}_q^{k}\right\rangle-\frac{\gamma_q(\rho)}{2}\|\mathbf{x}_q^{k+1}-\mathbf{x}_{q}^{k}\|_2^2\Big)\\
&\hspace{0.5cm}+{\rm Re}\left\langle\nabla_{\mathbf{x}_0}  L_{\rho}(\{\mathbf{x}_q^{k+1}\}_{q=1}^{Q}, \mathbf{x}_0^{k+1}, \mathbf{y}^{k}), \mathbf{x}_0^{k+1}-\mathbf{x}_0^{k}\right\rangle\\
&\hspace{4.5cm}-\frac{\gamma(\rho)}{2}\|\mathbf{x}_0^{k+1}-\mathbf{x}^{k}_0\|_2^2\\
&\stackrel{\rm (b)}\le -\sum_{q=1}^{Q}\frac{\gamma_q(\rho)}{2} \|\mathbf{x}^{k+1}_q-\mathbf{x}_{q}^{k}\|_2^2-\frac{\gamma(\rho)}{2} \|\mathbf{x}_0^{k+1}-\mathbf{x}_0^k\|_2^2.
\end{split}
\end{align}
Since $\nabla_{\mathbf{x}_q}^2 L_{\rho}(\{\mathbf{x}_q\}_{q=1}^{Q}, \mathbf{x}_0, \mathbf{y})  = (2^{q-1})^2\rho-\alpha_q > {0}$ always holds, which leads that $ L_{\rho}(\{\mathbf{x}_q\}, \mathbf{x}_0, \mathbf{y})$ is strongly convex w.r.t.\ each $\mathbf{x}_q$ with modulus $\gamma_q(\rho)$, we can obtain the properties of strong convexity \cite{boyd2004convex} such that
\begin{align}\label{eq:strconvXq}
&L_{\rho}(\mathbf{x}_q^k,\!\mathbf{x}_1^{k+1}\!\!\!\!,\!\cdots\!,\mathbf{x}_{q-1}^{k+1}\!, \mathbf{x}_{q+1}^k\!,\!\cdots\!,\!\mathbf{x}_Q^k,\!\mathbf{x}_0^{k},\! \mathbf{y}^{k})\nonumber\\
&\ge L_{\rho}(\mathbf{x}_q^{k+1},\!\mathbf{x}_1^{k+1}\!\!\!\!,\!\cdots\!,\mathbf{x}_{q-1}^{k+1}\!, \mathbf{x}_{q+1}^k\!,\!\cdots\!,\!\mathbf{x}_Q^k,\!\mathbf{x}_0^{k},\! \mathbf{y}^{k})\nonumber\\
&+{\rm Re}\left\langle\nabla_{\mathbf{x}_q}  L_{\rho}(\mathbf{x}_q^{k+1},\!\mathbf{x}_1^{k+1}\!\!\!\!,\!\cdots\!,\mathbf{x}_{q-1}^{k+1}\!, \mathbf{x}_{q+1}^k\!,\!\cdots\!,\!\mathbf{x}_Q^k,\!\mathbf{x}_0^{k},\! \mathbf{y}^{k}), \right.\nonumber\\
&\hspace{0.4cm}\left.\mathbf{x}_q^{k}-\mathbf{x}_q^{k+1}\right\rangle+\frac{\gamma_q(\rho)}{2}\|\mathbf{x}_q^{k}-\mathbf{x}_{q}^{k+1}\|_2^2, \nonumber\\
&\hspace{1.2cm}\  \forall~{\gamma_q(\rho)}\in (0,(2^{q-1})^2\rho-\alpha_q].
\end{align}
Since $\mathbf{H}^{H} \mathbf{H} $ is a positive definite matrix and $\rho>0$, $ \nabla_{\mathbf{x}_0}^2 L_{\rho}(\{\mathbf{x}_q\}_{q=1}^{Q}, \mathbf{x}_0, \mathbf{y})\succeq(\lambda_{\rm min}(\mathbf{H}^{H} \mathbf{H})+\rho)\mathbf{I}\succ0$ always holds, which leads that $ L_{\rho}(\{\mathbf{x}_q\}_{q=1}^{Q}, \mathbf{x}_0, \mathbf{y})$ is strongly convex w.r.t.\ $\mathbf{x}_0$ with modulus $\gamma(\rho)$,
we can obtain the properties of strong convexity \cite{boyd2004convex} such that
\begin{align}\label{eq:strconvX0}
&L_{\rho}(\{\mathbf{x}_q^{k+1}\}_{q=1}^{Q}, \mathbf{x}_0^{k}, \mathbf{y}^{k})\ge L_{\rho}(\{\mathbf{x}_q^{k+1}\}_{q=1}^{Q}, \mathbf{x}_0^{k+1}, \mathbf{y}^{k})\nonumber\\
&\hspace{1.2cm}+{\rm Re}\left\langle\nabla_{\mathbf{x}_0}  L_{\rho}(\{\mathbf{x}_q^{k+1}\}_{q=1}^{Q}, \mathbf{x}_0^{k+1}, \mathbf{y}^{k}), \mathbf{x}_0^{k}-\mathbf{x}_0^{k+1}\right\rangle\nonumber\\
&\hspace{1.2cm}+\frac{\gamma(\rho)}{2}\|\mathbf{x}_0^{k}-\mathbf{x}^{k+1}_0\|_2^2,\nonumber \\
&\hspace{1.2cm} \forall~{\gamma(\rho)}\in (0,{\lambda_{\rm min}(\mathbf{H}^{H} \mathbf{H})+\rho}].
\end{align}

In \eqref{eq:extra1}, the inequality ``$\stackrel{\rm (a)}{\le}$'' holds since we have used the fact \eqref{eq:strconvXq} and \eqref{eq:strconvX0};
the inequality ``$\stackrel{\rm (b)}{\le}$'' holds since we have used the optimality condition \cite{bertsekas2009convex} of subproblem \eqref{eq:x_q_update} and \eqref{eq:x0_update} , i.e., $\left\langle\nabla_{\mathbf{x}_q} L_{\rho}(\mathbf{x}_q^{k+1},\!\mathbf{x}_1^{k+1}\!\!\!\!,\!\cdots\!,\mathbf{x}_{q-1}^{k+1}\!, \mathbf{x}_{q+1}^k\!,\!\cdots\!,\!\mathbf{x}_Q^k,\!\mathbf{x}_0^{k},\! \mathbf{y}^{k}), \mathbf{x}_q^{k}\!-\!\mathbf{x}_q^{k+1}\right\rangle\\ \ge 0$ and $\nabla_{\mathbf{x}_0}  L_{\rho}(\{\mathbf{x}_q^{k+1}\}_{q=1}^{Q}, \mathbf{x}_0^{k+1}, \mathbf{y}^{k})= 0 $.

Combining the above two inequalities \eqref{eq:extra} and \eqref{eq:extra1}, we obtain
\begin{align*}%\label{eq:L_descent1}
&L_{\rho}(\{\mathbf{x}_q^{k+1}\}_{q=1}^{Q}, \mathbf{x}_0^{k+1}, \mathbf{y}^{k+1})-L_{\rho}(\{\mathbf{x}_q^{k}\}_{q=1}^{Q}, \mathbf{x}_0^{k}, \mathbf{y}^{k})\\%\nonumber
&\hspace{2cm}\le \sum_{q=1}^{Q}-\frac{\gamma_q(\rho)}{2} \|\mathbf{x}_q^{k+1}-\mathbf{x}_{q}^{k}\|_2^2\\
&\hspace{2cm}\quad\  -\Big(\frac{\gamma(\rho)}{2}-\frac{\lambda_{\rm max}^2(\mathbf{H}^{H} \mathbf{H})}{\rho}\Big) \|\mathbf{x}_0^{k+1}-\mathbf{x}_0^{k}\|_2^2.%\nonumber
\end{align*}
The desired result is obtained.

Since $\rho{\gamma(\rho)}$ is monotonically increasing w.r.t $\rho$ and ${\lambda_{\rm max}^2(\mathbf{H}^{H} \mathbf{H})}$ is a constant,
suppose $\rho > \sqrt{2} \lambda_{\rm max}(\mathbf{H}^H\mathbf{H})$ is satisfied, which leads to the following inequality holds.
\begin{align}
&\Big(\frac{\gamma(\rho)}{2}-\frac{\lambda_{\rm max}^2(\mathbf{H}^{H} \mathbf{H})}{\rho}\Big)>0,\nonumber \\
&\hspace{2.5cm}\forall~{\gamma(\rho)}\in (0,{\lambda_{\rm min}(\mathbf{H}^{H} \mathbf{H})+\rho}].\nonumber
\end{align}

The above result implies that the value of the augmented Lagrangian function $L_{\rho}(\{\mathbf{x}_q^{k}\}_{q=1}^{Q}, \mathbf{x}_0^{k}, \mathbf{y}^{k} )$ will always decrease if $4^{q-1}\rho>\alpha_q$ and $\rho > \sqrt{2} \lambda_{\rm max}(\mathbf{H}^H\mathbf{H})$ are satisfied.
$\hfill\blacksquare$
\end{proof}

\begin{lemma}\label{lemma:L_bounded1}
Assume parameters $\{\alpha_q\}_{q=1}^{Q}$ and $\rho$ satisfy $4^{q-1}\rho>\alpha_q$ and $\rho > \sqrt{2} \lambda_{\rm max}(\mathbf{H}^H\mathbf{H})$. Let $\{\{\mathbf{x}_q^{k}\}_{q=1}^{Q},\mathbf{x}_0^k,\mathbf{y}^k\}$ be generated by Algorithm 1, then the following limit exists and is bounded from below
\begin{align}
\lim_{k\to\infty} L_{\rho}(\{\mathbf{x}_q^{k}\}_{q=1}^{Q},\mathbf{x}_0^k,\mathbf{y}^k) >-\infty.
\end{align}
\end{lemma}
\begin{proof}
We have the following series of inequalities
\begin{align}\label{eq:Lbound}
&L_{\rho}(\{\mathbf{x}_q^{k+1}\}_{q=1}^{Q}, \mathbf{x}_0^{k+1}, \mathbf{y}^{k+1})\nonumber\\
&=\sum_{q=1}^{Q}\big(-\frac{\alpha_q}{2}\big\Vert\mathbf{x}_q^{k+1} \big \Vert_2 ^{2}\big )+\ell\left(\mathbf{x}_0^{k+1}\right)\nonumber\\
&\ \ \ +{\rm Re}\Big\langle \mathbf{x}_0^{k+1}-\sum_{q=1}^{Q}2^{q-1} \mathbf{x}_q^{k+1}, \mathbf{y}^{k+1}\Big\rangle\nonumber\\
&\quad+\frac{\rho}{2}\Big\|\mathbf{x}_0^{k+1}-\sum_{q=1}^{Q}2^{q-1}\mathbf{x}_q^{k+1}\Big\|_2^2 \nonumber\\
&\stackrel{\rm (a)}=\sum_{q=1}^{Q}\big(-\frac{\alpha_q}{2}\Vert\mathbf{x}_q^{k+1} \Vert_2 ^{2}\big)+\ell\left(\mathbf{x}_0^{k+1}\right)\nonumber\\
&\quad +{\rm Re}\Big\langle \sum_{q=1}^{Q}2^{q-1} \mathbf{x}_q^{k+1}-\mathbf{x}_0^{k+1}, \nabla_{\mathbf{x}_0} \ell\left(\mathbf{x}_0^{k+1}\right)\Big\rangle\nonumber\\
&\quad +\frac{\rho}{2}\Big\|\mathbf{x}_0^{k+1}-\sum_{q=1}^{Q}2^{q-1}\mathbf{x}_q^{k+1}\Big\|_2^2 \nonumber\\
&\stackrel{\rm (b)}\ge \sum_{q=1}^{Q}\big(-\frac{\alpha_q}{2}\Vert\mathbf{x}_q^{k+1} \Vert_2 ^{2}\big)+\ell\big(\sum_{q=1}^{Q}2^{q-1} \mathbf{x}_q^{k+1}\big)\nonumber\\
&\quad +\frac{\rho-\lambda_{\rm max}(\mathbf{H}^H\mathbf{H})}{2}\Big\|\mathbf{x}_0^{k+1}-\sum_{q=1}^{Q}2^{q-1}\mathbf{x}_q^{k+1}\Big\|_2^2.
\end{align}
In \eqref{eq:Lbound},  the equality ``$\stackrel{\rm (a)}{=}$'' holds since we have used \eqref{eq:z_nabla_expression}.
Since we show that gradient $\|\nabla_{\mathbf{x}_0}\ell\left(\mathbf{x}_0\right)\|_2$ is Lipschitz continuous in Lemma \ref{lemma:z1}, according to the Decent Lemma \cite{bertsekas1999nonlinear} and $ \|\nabla_{\mathbf{x}_0}^2 \ell\left(\mathbf{x}_0\right)\|_2 \le \lambda_{\rm max}(\mathbf{H}^H\mathbf{H})$ \cite{boyd2004convex}, we can obtain
\begin{align*}
&\ell\big(\sum_{q=1}^{Q}2^{q-1} \mathbf{x}_q^{k+1}\big)\le \ell\left(\mathbf{x}_0^{k+1}\right)\\
&+{\rm Re}\Big\langle \nabla_{\mathbf{x}_0} \ell\left(\mathbf{x}_0^{k+1}\right),\sum_{q=1}^{Q}2^{q-1} \mathbf{x}_q^{k+1}-\mathbf{x}_0^{k+1} \Big\rangle \\
&+\frac{\lambda_{\rm max}(\mathbf{H}^H\mathbf{H})}{2}\Big\|\sum_{q=1}^{Q}2^{q-1}\mathbf{x}_q^{k+1}-\mathbf{x}_0^{k+1}\Big\|_2^2.
\end{align*}
This implies the inequality ``$\stackrel{\rm (b)}{\ge}$'' in \eqref{eq:Lbound} holds true.
Since $\sum_{q=1}^{Q}(-\frac{\alpha_q}{2}\Vert\mathbf{x}_q^{k+1} \Vert_2 ^{2})+\ell\big(\sum_{q=1}^{Q}2^{q-1} \mathbf{x}_q^{k+1}\big)$ is bounded over $\mathbf{x}_{qR}, \mathbf{x}_{qI} \in [-1\ 1]^U$, as well as the fact that $\rho-\lambda_{\rm max}(\mathbf{H}^H\mathbf{H})>0$ comes from $\rho> \sqrt{2} \lambda_{\rm max}(\mathbf{H}^H\mathbf{H})$. Using these two cases leads to the desired result that $L_{\rho}(\{\mathbf{x}_q^{k}\}_{q=1}^{Q},\mathbf{x}_0^k,\mathbf{y}^k)$ is lower bounded and Lemma \ref{lemma:L_bounded1} has proved.

Combining Lemma \ref{lemma:L_difference} and Lemma \ref{lemma:L_bounded1}, we conclude that $L_{\rho}(\{\mathbf{x}_q^{k}\}_{q=1}^{Q},\mathbf{x}_0^k,\mathbf{y}^k)$ is monotonically decreasing and is convergent.
$\hfill\blacksquare$
\end{proof}

\section{Proof of Theorem 1}\label{PS-ADMM Proof}

First, we prove \eqref{convergence variables} in Theorem \ref{thm:convergence}.
Since Lemma \ref{lemma:L_difference} holds, we sum both sides of the inequality \eqref{eq:L_difference} when $k=1,2,\cdots,+\infty$ and obtain
    \begin{align*}
    &L_{\rho}\left (\{\mathbf{x}_q^{1}\}_{q=1}^{Q}, \mathbf{x}_0^{1}, \mathbf{y}^{1} \right)-L_{\rho}\left (\{\mathbf{x}_q^{k}\}_{q=1}^{Q}, \mathbf{x}_0^{k}, \mathbf{y}^{k}\right)\\
      &\geq \sum_{k=1}^{+\infty}\sum_{q=1}^{Q}\frac{\gamma_q(\rho)}{2}\|\mathbf{x}_q^{k+1}-\mathbf{x}_{q}^{k}\|_2^2 \\
&+\sum_{k=1}^{+\infty}\Big(\frac{\gamma(\rho)}{2}-\frac{\lambda_{\rm max}^2(\mathbf{H}^{H} \mathbf{H})}{\rho}\Big)\|\mathbf{x}_0^{k+1}-\mathbf{x}_0^k\|_2^2.\\
  \end{align*}
Using Lemma \ref{lemma:L_bounded1}, the above inequality indicates that summation of infinite positive terms is less than some constant. Therefore, we can obtain \eqref{eq:differenceX0} and \eqref{eq:differenceXq}.
\begin{equation}\label{eq:differenceX0}
 \lim\limits_{k\rightarrow+\infty} \|\mathbf{x}_0^{k+1}-\mathbf{x}_0^k\|_2 = 0.
 \end{equation}
\begin{equation}\label{eq:differenceXq}
 \lim\limits_{k\rightarrow+\infty}\|\mathbf{x}_q^{k+1}-\mathbf{x}_{q}^{k}\|_2 = 0,  \hspace{0.3cm} \forall~\mathbf{x}_q\in \tilde{\mathcal{X}}_q^{U}, q=1,2,\cdots,Q.
\end{equation}
Plugging \eqref{eq:differenceX0} into \eqref{eq:y_difference1}'s right side, we can get
\begin{equation}\label{limYzero}
\lim\limits_{k\rightarrow+\infty}\|\mathbf{y}^{k+1}-\mathbf{y}^k\|_2 = 0.
\end {equation}
Combining \eqref{limYzero} and \eqref{eq:y_update}, we further have
\begin{align}\label{eq:differenceX01}
\lim\limits_{k\rightarrow+\infty}\|\mathbf{x}_0^{k+1}-\sum_{q=1}^{Q}2^{q-1}\mathbf{x}_q^{k+1} \|_2 = 0.
\end{align}
Since $\mathbf{x}_{qR}, \mathbf{x}_{qI} \in [-1\ 1]^U$, we can obtain the following convergence results from \eqref{eq:differenceXq}.
\begin{equation}\label{convergence xq}
\lim\limits_{k\rightarrow+\infty}\mathbf{x}_q^{k}\!=\!\mathbf{x}_q^{*},~\forall~q=1,2,\cdots,Q.
\end{equation}
Plugging \eqref{convergence xq} into \eqref{eq:differenceX01}, we can conclude that $\mathbf{x}^{k}_0$ is bounded and have a limit point
\begin{equation}\label{convergence_x0}
\lim\limits_{k\rightarrow+\infty}\mathbf{x}_0^{k}=\mathbf{x}_0^{*}=\sum_{q=1}^{Q}2^{q-1}\mathbf{x}_q^{*}.
 \end{equation}

From \eqref{eq:z_nabla_expression},  we can derive
  \begin{equation}\label{limy_x0}
    \lim\limits_{k\rightarrow+\infty} \mathbf{y}^k = \lim\limits_{k\rightarrow+\infty}-\nabla_{\mathbf{x}_0} \ell\left(\mathbf{x}_0^{k}\right).
   \end{equation}
Since $\|\nabla_{\mathbf{x}_0} \ell\left(\mathbf{x}_0^{k+1}\right)-\nabla_{\mathbf{x}_0} \ell\left(\mathbf{x}_0^{k}\right)\|_2^2 \le \lambda_{\rm max}^2(\mathbf{H}^{H} \mathbf{H}) \|\mathbf{x}_0^{k+1}-\mathbf{x}_0^{k}\|_2^2$ and $\mathbf{x}^{k}_0$ is bounded, we can conclude that all the elements in $\nabla_{\mathbf{x}_0} \ell\left(\mathbf{x}_0\right)$ are bounded. From \eqref{limy_x0}, it indicates that $\mathbf{y}^k$ is also bounded. Combining this result with \eqref{limYzero}, we can get
\begin{equation}\label{convergence z}
\lim\limits_{k\rightarrow+\infty} \mathbf{y}^k =  \mathbf{y}^{*}.
\end{equation}
The above result indicates that there exists limit point $(\{\mathbf{x}_q^*\}_{q=1}^{Q}, \mathbf{x}_0^*, \mathbf{y}^*)$ for the sequences $\{\{\mathbf{x}_q^{k}\}_{q=1}^{Q}, \mathbf{x}_0^{k}, \mathbf{y}^{k}\}$ of iterations generated by the proposed Algorithm 1.

Second, we prove $\{\mathbf{x}_q^*\}_{q=1}^{Q}$ is a stationary point of original problem \eqref{eq:PS_ML}.

Since $\{\mathbf{x}_q^{k+1}\}_{q=1}^{Q} = \underset{\mathbf{x}_q \in \tilde{\mathcal{X}}_q^{U}}{\arg\min} \; L_{\rho}\left (\{\mathbf{x}_q\}_{q=1}^{Q}, \mathbf{x}^{k}_0, \mathbf{y}^{k} \right) $ in \eqref{eq:x_q_update}, and $ L_{\rho}\left (\{\mathbf{x}_q\}_{q=1}^{Q}, \mathbf{x}^{k}_0, \mathbf{y}^{k} \right)$ is strongly convex w.r.t.\ $\mathbf{x}_q$, we have the following optimality conditions.

\begin{equation}\label{optimality_xq}
\begin{split}
&{\rm Re}\Big\langle\nabla_{\mathbf{x}_q}\Big(\ell\left(\mathbf{x}_0^{k}\right)-\sum_{q=1}^{Q}\frac{\alpha_q}{2}\Vert\mathbf{x}_q^{k+1} \Vert_2 ^{2}\\
&+\big\langle \mathbf{x}_0^{k}-\sum_{q=1}^{Q}2^{q-1} \mathbf{x}_q^{k+1}, \mathbf{y}^{k}\big\rangle+\frac{\rho}{2}\big\|\mathbf{x}_0^{k}-\sum_{q=1}^{Q}2^{q-1}\mathbf{x}_q^{k+1} \big\|_2^2\Big),\\
&\hspace{0.1cm} \mathbf{x}_q-\mathbf{x}_q^{k+1} \Big\rangle \ge 0,  \quad \forall~\mathbf{x}_q\in \tilde{\mathcal{X}}_q^{U}, q=1,2,\cdots,Q.
\end{split}
\end{equation}
When $k\rightarrow+\infty$, plugging convergence results \eqref{convergence_x0} into \eqref{optimality_xq}, since $\mathbf{y}^k$ and $\rho$ are bounded, we can drop two terms in \eqref{optimality_xq} and obtain
\begin{align*}
& {\rm Re}\Big \langle\nabla_{\mathbf{x}_q} \Big(\ell\big(\sum_{q=1}^{Q}2^{q-1} \mathbf{x}_q^*\big)-\sum_{q=1}^{Q}\frac{\alpha_q}{2}\Vert\mathbf{x}_q^* \Vert_2 ^{2}\Big),\mathbf{x}_q-\mathbf{x}^*_q\Big\rangle\nonumber\\
& \quad\quad\ge 0,\hspace{1.8cm}\forall~\mathbf{x}_q\in \tilde{\mathcal{X}}_q^{U},\; q=1,2,\cdots,Q.
\end{align*}
which completes the proof.$\hfill\blacksquare$

\section{Proof of Theorem 2}\label{Iteration complexity Proof}
To be clear, here we rewrite \eqref{eq:L_difference} as
\[
  \begin{split}
    &L_{\rho}\left (\{\mathbf{x}_q^{k}\}_{q=1}^{Q}, \mathbf{x}_0^{k}, \mathbf{y}^{k} \right)-L_{\rho}\left (\{\mathbf{x}_q^{k+1}\}_{q=1}^{Q}, \mathbf{x}_0^{k+1}, \mathbf{y}^{k+1}\right)\\
      &\geq \sum_{q=1}^{Q}-\frac{\gamma_q(\rho)}{2}\|\mathbf{x}_q^{k+1}-\mathbf{x}_{q}^{k}\|_2^2 \nonumber\\
&-\Big(\frac{\gamma(\rho)}{2}-\frac{\lambda_{\rm max}^2(\mathbf{H}^{H} \mathbf{H})}{\rho}\Big)\|\mathbf{x}_0^{k+1}-\mathbf{x}_0^k\|_2^2.\\
  \end{split}
\]

According to Lemma \ref{lemma:L_difference}, there exists a constant $C\!=\!\min\!\left\{\!\{ \frac{\gamma_q(\rho)}{2}\}_{q=1}^{Q}, \Big(\frac{\gamma(\rho)}{2}\!\!-\!\!\frac{\lambda_{\rm max}^2(\mathbf{H}^{H} \mathbf{H})}{\rho}\Big)\! \right\}$
such that
\[
  \begin{split}
    &L_{\rho}\left (\{\mathbf{x}_q^{k}\}_{q=1}^{Q}, \mathbf{x}_0^{k}, \mathbf{y}^{k} \right)-L_{\rho}\left (\{\mathbf{x}_q^{k+1}\}_{q=1}^{Q}, \mathbf{x}_0^{k+1}, \mathbf{y}^{k+1}\right)\\
      &\geq C\Big(\sum_{q=1}^{Q}\|\mathbf{x}_q^{k+1}-\mathbf{x}_{q}^{k}\|_2^2 +\|\mathbf{x}_0^{k+1}-\mathbf{x}_0^k\|_2^2\Big).\\
  \end{split}
\]
 Summing both sides of the above inequality from $k=1,\cdots, K$, we have
 \begin{equation}\label{dL}
   \begin{split}
      &L_{\rho}\left (\{\mathbf{x}_q^{1}\}_{q=1}^{Q}, \mathbf{x}_0^{1}, \mathbf{y}^{1} \right)-L_{\rho}\left (\{\mathbf{x}_q^{K+1}\}_{q=1}^{Q}, \mathbf{x}_0^{K+1}, \mathbf{y}^{K+1}\right)\\
      &\geq \sum_{k=1}^{K}\Bigg(C\Big(\sum_{q=1}^{Q}\|\mathbf{x}_q^{k+1}-\mathbf{x}_{q}^{k}\|_2^2 +\|\mathbf{x}_0^{k+1}-\mathbf{x}_0^k\|_2^2\Big)\Bigg).\\
  \end{split}
 \end{equation}
 Since $t = \underset{k}{\rm min}\{k|\sum_{q=1}^{Q}\|\mathbf{x}_q^{k+1}-\mathbf{x}_{q}^{k}\|_2^2 + \|\mathbf{x}_0^{k+1}-\mathbf{x}_0^{k}\|_2^2\leq\epsilon\}$, we can change \eqref{dL} to
 \begin{equation}\label{Relax L}
   \begin{split}
    & L_{\rho}\left (\{\mathbf{x}_q^{1}\}_{q=1}^{Q}, \mathbf{x}_0^{1}, \mathbf{y}^{1} \right)-L_{\rho}\left (\{\mathbf{x}_q^{K+1}\}_{q=1}^{Q}, \mathbf{x}_0^{K+1}, \mathbf{y}^{K+1}\right)\\
    &  \geq tC\epsilon.
   \end{split}
 \end{equation}
 Since we have $L \!\left (\{\mathbf{x}_q^{K+1}\}_{q=1}^{Q}, \mathbf{x}_0^{K+1}, \mathbf{y}^{K+1}\right)\! \geq\! L\! \left (\{\mathbf{x}_q^*\}_{q=1}^{Q}, \mathbf{x}_0^*, \mathbf{y}^*\right)$, \eqref{Relax L} can be reduced to
 \[
   \begin{split}
     t \!&\leq\! \frac{1}{C\epsilon}\bigg(\!L_{\rho}\left (\{\mathbf{x}_q^{1}\}_{q=1}^{Q}, \mathbf{x}_0^{1}, \mathbf{y}^{1} \right) \!-\! L_{\rho}\left (\{\mathbf{x}_q^*\}_{q=1}^{Q}, \mathbf{x}_0^*, \mathbf{y}^*\right)\!\bigg),
   \end{split}
 \]
 where $L_{\rho}\left (\{\mathbf{x}_q^*\}_{q=1}^{Q}, \mathbf{x}_0^*, \mathbf{y}^*\right)=\ell\left(\mathbf{x}^*\right) - \sum_{q=1}^{Q}\frac{\alpha_q}{2} \Vert\mathbf{x}_q^*\Vert_2 ^{2}$, which concludes the proof of Theorem \ref{iteration complexity}.
 $\hfill\blacksquare$

\ifCLASSOPTIONcaptionsoff
  \newpage
\fi

\end{document}